\documentclass[journal]{IEEEtran}

\usepackage{mathrsfs}
\usepackage{amsfonts,amsmath,amssymb}
\usepackage{dsfont}
\usepackage{cite}
\usepackage{graphicx}
 \usepackage{booktabs}

\usepackage{makecell}

\usepackage{caption}
\usepackage{bm}
\usepackage{bbding}
\usepackage{amsmath}
\usepackage{enumerate}
\usepackage{multirow}
\usepackage[ruled]{algorithm2e}
\usepackage{algorithmicx}
\usepackage{xcolor}
\usepackage{algpseudocode}
\usepackage{algcompatible}
\usepackage{subcaption}
\usepackage{array}
\usepackage{slashbox}
\usepackage{CJK}
\usepackage{graphicx}

\usepackage{url}
\usepackage{amsthm}
\newtheorem{theorem}{Theorem}

 \newtheorem{lemma}{Lemma}

  \newtheorem{remark}{Remark}

  \newcommand{\figref}[1]{\figurename~\ref{#1}}

\begin{document}
\title{Jittering Effects Analysis and Beam Training Design for UAV Millimeter Wave Communications}
\author{Wei~Wang,~\IEEEmembership{Member,~IEEE}, and~Wei~Zhang,~\IEEEmembership{Fellow,~IEEE}
\thanks{
W. Wang and W. Zhang are with School of Electrical Engineering and Telecommunications, The University of New South Wales, Sydney, Australia (e-mail: wei.wang@unsw.edu.au; w.zhang@unsw.edu.au).}
}

\maketitle

\begin{abstract}
Jittering effects significantly degrade the performance of  UAV millimeter-wave (mmWave) communications. To investigate the impacts of UAV jitter on mmWave communications, we firstly model UAV mmWave channel based on the geometric relationship between element antennas of the uniform planar arrays (UPAs). Then, we extract the relationship between (I) UAV attitude angles \& position coordinates and (II) angle of arrival (AoA)  \& angle of departure (AoD) of mmWave channel, and we also derive the distribution of AoA/AoD at UAV side from the random fluctuations of UAV attitude angles, i.e., UAV jitter. In beam training design, with the relationship between attitude angles and AoA/AoD, we propose to generate a rough estimate of AoA and AoD from UAV navigation information.  Finally, with the rough AoA/AoD estimate, we develop a compressed sensing (CS) based beam training scheme with constrained sensing range as the fine AoA/AoD estimation. Particularly, we construct a partially random sensing matrix to narrow down the sensing range of CS-based beam training. Numerical results show that our proposed UAV beam training scheme assisted by navigation information can achieve better accuracy  with reduced training length in AoA/AoD estimation and  is thus more suitable for UAV mmWave communications under jittering effects.
\end{abstract}


\maketitle

\section{Introduction}

Millimeter-wave (mmWave) communication, with multi-gigahertz bandwidth available, is a promising communication technology for the future wireless communications. However, signals over mmWave frequency band (30-300GHz) suffer from severe free-space propagation loss \cite{hemadeh2017millimeter, mmWaveRef, MammWave}. Massive multiple input multiple output (MIMO) and beamforming that focuses the wireless signal towards a specific direction are envisioned as an enabling technology for mmWave communications \cite{shen2015downlink, OptBeamPattern}. The shortcoming of beamforming in terrestrial mmWave communication network is its vulnerability to blockage \cite{Blockage}. Fortunately, when deployed in aerial platforms, e.g.,  unmanned aerial vehicle (UAV),  the mounted mmWave communication system can dynamically adjust its location in the three-dimensional space so as to guarantee line-of-sight (LoS) links and provide high-throughput services \cite{lin2018sky, zhang2019research, liu2019safeguarding, zhang2020multi}.  Owing to its mobility and high throughput,  UAV mmWave communications have been proposed as an on-demand solution for high-capacity wireless backhaul in cellular networks \cite{bertizzolo2019mmbac}.

As a type of small aircrafts, the key challenge in the use of UAVs is their high sensitivity to wind. Wind gust, a.k.a. turbulence, will result in brief changes of wind velocity from a steady value caused by changes in atmospheric pressure and temperature \cite{WindMea,moyano2013quadrotor}. Although estimation of wind disturbances has been extensively studied in flight control to mitigate wind effect \cite{arain2014real,perozzi2018wind},    UAV jitter/vibration, i.e., the unintended high-frequency change of UAV attitude/orientation, is still non-negligible, especially for multi-rotor UAVs. In addition, UAV jitter is also contributed by the high-frequency vibration from the propellers and rotors \cite{li2017development}, which is inherently unavoidable for small aircrafts. Jittering effect, i.e., the fast variation of UAV communication channels caused by UAV jitter, is a key character that make UAV communications significantly different from terrestrial communications. For wireless communication systems deployed in UAV platform, jittering effects could severely deteriorate the quality of wireless channel when carrier frequency is high, as the path variation caused by UAV jitter is significant compared with the wavelength \cite{banagar2020impact}. In \cite{yan2019comprehensive},  propagation attenuation,  shadowing, root mean square (RMS), and  delay spread (DS) of  UAV channel over lower frequency band, typically ranging from $0.8$GHz to $8$GHz, are studied. In \cite{dabiri2020analytical}, the impacts of UAV jitter on the  signal to interference noise ratio (SINR) of UAV mmWave channel are analytically studied. In \cite{sanchez2020robust}, experimental analysis of UAV hovering effects, which refer to the continuous motion around the target coordinates caused by inaccurate GPS signals, is performed for UAV mmWave communications at 60GHz. However, \cite{yan2019comprehensive,dabiri2020analytical, sanchez2020robust} mainly investigate the aggregate impacts of UAV jitter on the wireless channel, i.e., propagation attenuation, shadowing, RMS, DS, and SINR, while the spatial domain parameters, e.g., AoA and AoD, are more worthy of study in mmWave communications, as mmWave communications are characterized by large scale antenna array and directional transmission.

In \cite{xu2020multiuser, xu2018robust},  trajectory design and resource allocation are performed for UAV communications under the imperfect knowledge of AoA/AoD caused by jittering effects and the user location uncertainties caused by inaccurate GPS signal.  In \cite{wu2020secrecy, WuA2G2020}, the impacts of  of UAV jitter on UAV secure communications are modeled and analyzed, and, based on the proposed model, UAV jitters are exploited to enhance the secrecy performance. Although \cite{xu2020multiuser, xu2018robust,wu2020secrecy, WuA2G2020} have captured the characteristics of  jittering effects to wireless channels by modeling it as the small-scale fluctuations of AoA/AoD around the authentic value, the proposed models cannot reflect the connection between UAV jitter (variations of attitude/orientation angles) and the resulting variation of AoA/AoD. In \cite{yuan2020learning}, UAV jitter is assumed to impact merely the pitch angle of UAV, and the variation of pitch angle is further assumed to impact only the altitude angle (which is considered to be equal to both AoA and AoD of phased array in \cite{yuan2020learning}). Apparently, this assumption is over-simplified. In fact, due to the randomness of airflow and aircraft body vibration, the vibrations of UAV attitude are in all directions but merely differ in the scale at different angles \cite{wu2020secrecy}.

Understanding the relationship between UAV attitude angles (i.e., yaw, pitch, and roll) and AoA/AoD of mmWave channel is important.  UAV navigation system consists of global positioning system (GPS), barometer, and inertial measurement unit (IMU), e.g., gyroscope and accelerometer, which can provide estimation of  UAV attitude (enabled by gyroscope and accelerometer) and 3D UAV position (enabled by GPS and barometer)\cite{DJI, bertizzolo2019mmbac}. With the relationship between UAV attitude angles and AoA/AoD, the random variation of AoA/AoD caused by UAV jitter can be modeled and analyzed. In addition,  we can also use the attitude angles and 3D position  yielded by  navigation system to obtain a rough estimate of AoA/AoD as the prior knowledge for the subsequent beam training design, even though UAV jitter will inevitably affect the accuracy of attitude angle measurement, and the 3-D position provided by civilian GPS and barometer is inherently not accurate enough. It is noteworthy that, as an electronic device,  the response time of mmWave communication systems is in a much smaller order of magnitude than  the inertial units. Thus, it is possible for the variation of AoA/AoD to be tracked and then compensated by an efficient and fast beam training design.

In this paper, we propose to analyze the jittering effects to UAV mmWave channel response and  customize a beam training design for UAV mmWave communications under jittering effects.  To build the connection between UAV jitter (the fluctuation of  attitude angles) and AoA/AoD fluctuation, we firstly extract the relationship between  AoA \& AoD and UAV attitude angles through geometric channel modeling, and based on which we have made the following contributions:
\begin{itemize}
  \item We analyze the jittering effects to UAV mmWave channel response and derive the distribution of the AoA/AoD at UAV side under jittering effects, and we find that the degree of the  AoA/AoD fluctuation at UAV side depends on both UAV position and the desired UAV attitude. To our best knowledge, this is the first work that analytically builds the connection between UAV jitter (i.e., attitude angles fluctuations) and its effects on mmWave channel (i.e., AoA/AoD fluctuations).
  \item Based on the relationship between AoA \& AoD and UAV attitude \& position, we study the impacts of UAV attitude estimation error and UAV position estimation error on navigation-system-yielded AoA \& AoD, and we find that the estimation error of UAV position brings about negligible AoA/AoD error at BS side, but the estimation error of UAV attitude results in significant  AoA/AoD error at UAV side.
  \item We propose a direction-constrained CS-based beam training scheme  to refine the rough AoA/AoD estimate yielded by  navigation system. Unlike conventional channel measurement using omnidirectional random sensing matrix, we develop a novel method to generate the partially random sensing matrix with direction-constrained sensing range.
\end{itemize}
Numerical results show that our proposed UAV beam training scheme assisted by navigation information can achieve better accuracy  with reduced training length in AoA/AoD estimation and  is thus more suitable for UAV mmWave communications under jittering effects

The rest of the paper is organized as follows. In Section II, we model UAV mmWave channel with jittering effects.  In Section III, we analyze the jittering effects to UAV mmWave channel response. In Section IV,
we propose a direction-constrained beam training scheme that is assisted by UAV navigation information. In Section V, numerical results are presented. Finally, in Section VI, we draw the conclusion.

{\em{Notations:\quad}} Column vectors (matrices) are denoted by bold-face lower (upper) case letters, $\mathbf{x}(n)$ denotes the $n$-th element in the vector $\mathbf{x}$, $(\cdot)^*$, $(\cdot)^T$ and $(\cdot)^{H}$  represent conjugate, transpose and  conjugate transpose operation, respectively.
$\mathbf{A}(i,j)$ refers to the $(i,j)$-th entry of the matrix $\mathbf{A}$.
Subtraction and addition of the  AoAs/AoDs are defined as $\Psi \ominus \Omega  \triangleq (\Psi - \Omega+1) \mod 2 -1$ and $\Psi \oplus \Omega  \triangleq (\Psi + \Omega+1) \mod 2 -1$ to guarantee the result is within the range $[-1,1)$, $\otimes$ denotes Kronecker product operation. In addition, some of the key parameters are listed in Table \ref{Parameters}.

\begin{table*}
 \centering
 \begin{tabular}{ c|c }
   \hline
  Symbol & Parameter \\
  \hline
  \hline
  $\mathbf{p}_{B}$ & The (absolute) position of BS \\
  \hline
  $\mathbf{p}_{B, \iota}$ & The (absolute) position of the $\iota$-th element antenna in BS \\
  \hline
  $\mathbf{a}_{B, \iota}$ & The relative position of the $\iota$-th element antenna in BS \\
  \hline
  $\mathbf{p}_{U}$ & The (absolute) position of UAV \\
  \hline
 $\mathbf{p}_{B,  \kappa}$ & The (absolute) position of the $\kappa$-th element antenna in UAV \\
  \hline
 $\mathbf{a}_{U, \kappa}$ & The relative position of   the $\kappa$-th element antenna in UAV \\
 \hline
 $\mathbf{R}$ & The rotation matrix of UAV \\
 \hline
  $\alpha$, $\beta$, $\gamma$ & Yaw angle, pitch angle,  roll angle \\
  \hline
  \makecell[c]{$\mathbf{R}_{Yaw}(\alpha)$,  $ \mathbf{R}_{Pitch}(\beta)$, \\$\mathbf{R}_{Roll}(\gamma)$}  & The elemental rotation matrices about yaw angle, pitch angle and roll angle \\
  \hline
  $N_{B}$ & The number of element antennas of the UPA at BS side\\
  \hline
  $N_{B,x}$, $N_{B,z}$ & The number of element antennas along x-axis and z-axis in the UPA at BS sider\\
  \hline
  $N_{U}$ & The number of element antennas of the UPA at UAV side\\
  \hline
  $N_{U,x}$, $N_{U,y}$ & The number of element antennas along x-axis and y-axis in the UPA at UAV sider\\
  \hline
     $h_{\kappa, \iota}$ & \makecell[c]{The complex coefficient of the LoS path between the $\kappa$-th element antenna \\  at  UAV side and the $\iota$-th element antenna at BS side }  \\
   \hline
    $d_{BU}$  &   The distance between BS and UAV \\
   \hline
     $\mathbf{e}_{BU}$ & The direction vector between BS and UAV \\
  \hline
  $\phi$, $\theta$ & Elevation angle and  azimuth angle in spherical coordinate system\\
   \hline
    $c$, $f_c$, $\lambda $ & Speed of light, carrier frequency, and signal wavelength \\
   \hline
    $d_{\kappa, \iota}$ &  The  distance between the  $\kappa$-th antenna of UAV and the $\iota$-th antenna of BS\\
   \hline
     $\mathbf{H}$ & The channel response matrix between the UPA at BS side the UPA at UAV side\\
    \hline
     $\mathbf{v}_B $ & The array response vectors of the UPA at BS side \\
    \hline
    $\mathbf{v}_U $ & The array response vectors of the UPA at UAV side\\
    \hline
     $(\Psi_B $, $\Omega_B)$ & The two-dimensional cosine  AoA/AoD of the UPA at BS side\\
    \hline
    $(\Psi_U $, $\Omega_U)$ & The two-dimensional cosine  AoA/AoD of the UPA at UAV side\\
     \hline
     $(\hat\Psi_B $, $\hat\Omega_B)$ & Navigation-system-yielded two-dimensional cosine  AoA/AoD of the UPA at BS side\\
    \hline
    $(\hat\Psi_U $, $\hat\Omega_U)$ & Navigation-system-yielded two-dimensional cosine  AoA/AoD of the UPA at UAV side\\
    \hline
    $(\tilde\Psi_U $, $\tilde\Omega_U)$ & Beam-training-yielded  two-dimensional cosine AoA/AoD of the UPA at UAV side\\
    \hline
      $ \mathbf{m}_{U}$  & The receive beamforming vector at UAV side \\
    \hline
   $\mathbf{f}_B$ & The transmit beamforming vector at BS side \\
    \hline
     $\varphi_{n_U}$ &The phase coefficient of the $n_{U}$-th element antenna of the UPA at UAV side\\
               \hline
     $\zeta_{n_a}$ & \makecell[c]{The center angle of the $n_a$-th sub-array configured in the generation  \\ algorithm of sub-array based direction-constrained random sensing matrix}\\
   \hline
 \end{tabular}
 \caption{List of the key parameters used in this paper} \label{Parameters}
\end{table*}

\section{Geometric Channel Modeling for UAV MmWave Communications}
In this section, geometric channel modeling is performed for UAV mmWave channel  to extract the relationship between UAV attitude $\&$ position and AoA/AoD.
\subsection{Antenna Configuration}
\begin{figure}[tp]{
\begin{center}{\includegraphics[width=5.5cm ]{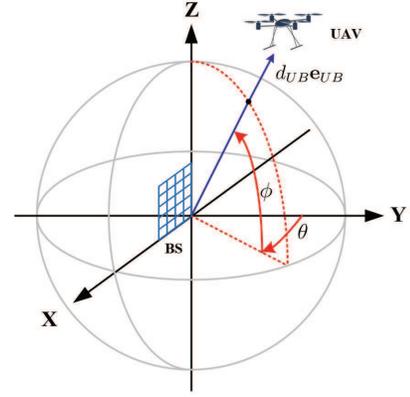}}
\caption{The link between UAV and BS, where $\phi$ is elevation angle, $\theta$ is azimuth angle}\label{UPABS}
\end{center}}
\end{figure}

We assume that a UPA with $N_{B}= N_{B,x} \times N_{B,z}$ element antennas is vertically mounted on the 2-D plane $y=0$ at BS side, which is shown in \figref{UPABS}. We denote the location of BS  as $\mathbf{p}_B = (0, 0, 0)^T$,   and thus the location of the $i$-th element antenna at BS side is represented as
\begin{align}
\mathbf{p}_{B, \iota} = \mathbf{p}_{B} + \mathbf{a}_{B,\iota}
\end{align}
where $\mathbf{a}_{B,\iota}$ is the relative position of the $\iota$-th element antenna with reference to the $1$-st element antenna (reference antenna)
\begin{align} \label{VecA}
\mathbf{a}_{B,\iota} = \frac{\lambda}{2}[x_{\iota}, 0, z_{\iota}]^T
\end{align}
where $\lambda$ is the wavelength of the carrier, $\frac{\lambda}{2}$  is the inter antenna spacing, the antenna indexes $x_\iota \in \{0, 1, \cdots, N_{B,x}-1\}$ and $z_\iota \in \{0, 1, \cdots, N_{B,z}-1 \}$, and the relative position of the reference antenna is $\mathbf{a}_{B,1} = [0, 0, 0]^T$.

Similarly, we assume that there is a UPA with $N_{U}= N_{U,x} \times N_{U,y}$ element antennas  mounted on the bottom of UAV, which is shown in \figref{UAVplatform}.
 We denote the location of UAV as $\mathbf{p}_U$, and thus the location of the $\kappa$-th element antenna at UAV side is
\begin{align}
\mathbf{p}_{U,\kappa} = \mathbf{p}_{U} + \underbrace{\mathbf{R}\mathbf{a}_{U,\kappa}}_{\bar{\mathbf{a}}_{U,\kappa}}
\end{align}
where $\mathbf{R}$  is the rotation matrix that depicts the attitude of UAV,

\begin{figure}[h]{
\begin{center}{\includegraphics[width=5cm ]{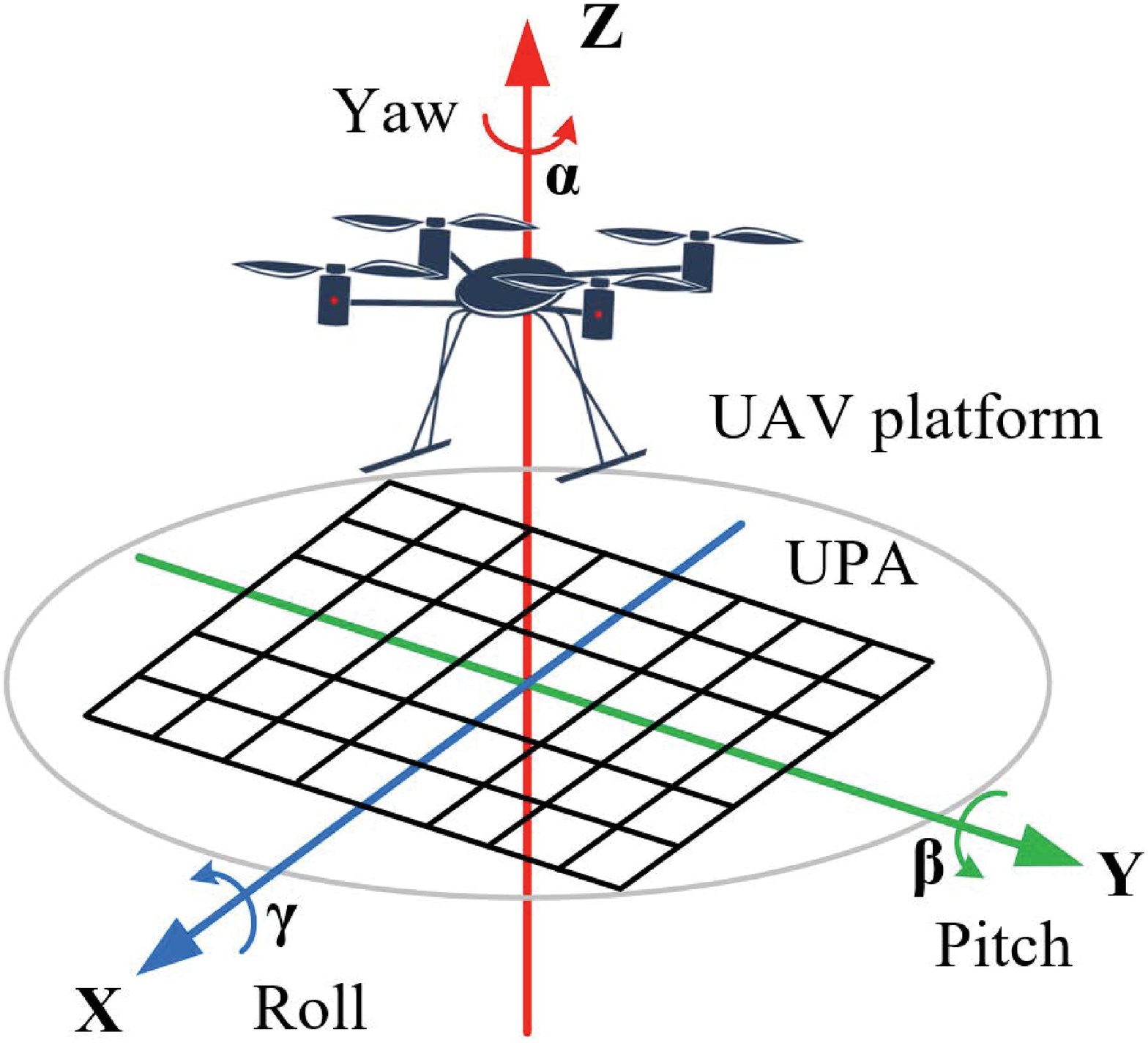}}
\caption{Yaw, pitch and roll rotations of a quadcopter UAV.  }\label{UAVplatform}
\end{center}}
\end{figure}

Note that the rotation matrix $\mathbf{R}$ is characterized by Euler angles, i.e., yaw, pitch and roll \cite{zhang2018measurement}, which is shown in \figref{UAVplatform}, and it is  written by
\begin{align}
\mathbf{R} =  \mathbf{R}_{Yaw}(\alpha) \mathbf{R}_{Pitch}(\beta) \mathbf{R}_{Roll}(\gamma)
\end{align}
where
\begin{align}
\mathbf{R}_{Yaw}(\alpha) &= \left(
                         \begin{array}{ccc}
                           \cos \alpha & -\sin \alpha & 0 \\
                           \sin \alpha & \cos \alpha & 0\\
                           0 & 0 & 1 \\
                         \end{array}
                       \right) \notag \\
\mathbf{R}_{Pitch}(\beta) &= \left(
                         \begin{array}{ccc}
                           \cos \beta & 0 & \sin \beta \\
                           0 & 1 & 0\\
                           -\sin \beta & 0 & \cos \beta \\
                         \end{array}
                       \right) \notag \\
\mathbf{R}_{Roll}(\gamma) &= \left(
                         \begin{array}{ccc}
                           1 & 0 & 0 \\
                           0 & \cos \gamma & -\sin\gamma\\
                           0 & \sin\gamma & \cos \gamma \\
                         \end{array}
                       \right) \notag
\end{align}
and $\mathbf{R}$ can be further expressed as \eqref{RotMatrix}.

\begin{figure*}
\begin{align} \label{RotMatrix}
\mathbf{R} =  \left(
                \begin{array}{ccc}
                  \cos \alpha \cos \beta & \cos \alpha \sin \beta \sin \gamma - \sin \alpha \cos \gamma & \cos\alpha\sin\beta\cos\gamma + \sin\alpha\sin\gamma \\
                  \sin\alpha \cos \beta & \sin\alpha\sin\beta\sin\gamma + \cos\alpha\cos\gamma & \sin\alpha\sin\beta\cos\gamma-\cos\alpha\sin\gamma \\
                  -\sin\beta & \cos\beta\sin\gamma & \cos\beta\cos\gamma \\
                \end{array}
              \right)
\end{align}
\hrulefill
\end{figure*}

$\mathbf{a}_{U,\kappa}$ is the relative position of the $\kappa$-th element antenna with reference to the $1$-st element antenna (reference antenna) when UPA is on the 2-D plane $z=0$, which is represented as
\begin{align} \label{VecA}
\mathbf{a}_{U,\kappa} = \frac{\lambda}{2}\left[x_\kappa-\frac{N_{U,x}-1}{2}, y_\kappa - \frac{N_{U,y}-1}{2}, 0\right]^T
\end{align}
where the antenna indexes $x_\kappa \in \{0, 1, \cdots, N_{U,x}-1\}$ and $y_\kappa \in \{0, 1, \cdots, N_{U,y}-1 \}$, and the relative position of the reference antenna is $\mathbf{a}_{U,1} = \frac{\lambda}{2}[ -\frac{N_{U,x}-1}{2}, - \frac{N_{U,y}-1}{2}, 0]^T$.

\subsection{Geometric Channel Modeling }

The channel response for mmWave MIMO with array antenna can be depicted using the geometry dependent model \cite{el2014spatially, hemadeh2017millimeter}, in which each propagation path is characterized by the array shape of Tx/Rx antenna and the spatial domain parameters, i.e., AoA/AoD. However, the exact relationship between AoA/AoD and the placement of antenna arrays (i.e., relative position and antenna orientation)  is rarely explored in wireless communication systems. The relationship is the prerequisite for modeling UAV jittering effects on mmWave channel.  In order to unravel the relationship, we start channel modeling from the element antennas at BS side and UAV side. According to the free-space path loss model in Eq. (2.7) and Eq. (2.8) of  \cite{goldsmith2005wireless}, the complex coefficient of the LoS path between the $\kappa$-th element antenna in  UAV side and the $\iota$-th element antenna at BS side is expressed as
\begin{align} \label{CHsingle}
h_{\kappa, \iota}=  \frac{\lambda  }{4\pi d_{\kappa, \iota}} e^{j\frac{-2\pi f_c d_{\kappa, \iota}}{c}} =   \frac{\lambda  }{4\pi d_{\kappa, \iota}}  e^{j\frac{-2\pi d_{\kappa, \iota}}{\lambda }}
\end{align}
where $c$ is the speed of light, $f_c$ is the carrier frequency, $\lambda = \frac{c}{f_c} $ is the signal wavelength and $d_{\kappa, \iota}$ is the distance between the  $\kappa$-th antenna of UAV and the $\iota$-th antenna of BS, which is represented as
\begin{align}
d_{\kappa, \iota} &=   \|   \mathbf{p}_{B,\iota} - \mathbf{p}_{U,\kappa} \|_2  = \|\mathbf{p}_B + \mathbf{a}_{B,\iota} -\mathbf{p}_U - \bar{\mathbf{a}}_{U,\kappa} \|_2
\end{align}
In \eqref{CHsingle}, the term $ \frac{1  }{d_{\kappa, \iota}} $ determines the magnitude of $h_{\kappa, \iota}$, and the term $d_{\kappa, \iota}$ determines the phase
of $h_{\kappa, \iota}$.

Through Taylor series expansion, we obtain the approximations of $\frac{1}{d_{\kappa, \iota}}$ and $d_{\kappa, \iota}$ as follows.
\begin{subequations}
\begin{align}
\frac{1}{d_{\kappa, \iota} } \approx & \frac{1}{\| \mathbf{p}_{B} - \mathbf{p}_{U}\|_2} \label{Approx1}\\
d_{\kappa, \iota}  \approx & \|\mathbf{p}_{B} -\mathbf{p}_{U}  \|_2 + (\mathbf{a}_{B, \iota} - \bar{\mathbf{a}}_{U,\kappa})^T  \frac{ \mathbf{p}_{B}-\mathbf{p}_{U}}{\|\mathbf{p}_{B} -  \mathbf{p}_{U} \|_2}  \label{Approx2}
\end{align}\label{Approx}
\end{subequations}
In \eqref{Approx1}, only the zeroth order component is conserved because $\mathbf{a}_{B,\iota} -\bar{\mathbf{a}}_{U,\kappa}$ is of millimeter-level and by contrast $ \mathbf{p}_{B} -  \mathbf{p}_{U} $ is of kilometer-level; In \eqref{Approx2},  as $\frac{d_{\kappa,\iota}}{\lambda}$ determines phase of $h_{\kappa, \iota}$, both zeroth and first order components are conserved, and second and higher order components are omitted due to their negligible impact on phase.

Substituting \eqref{Approx} into \eqref{CHsingle}, we have
\begin{align} \label{CHsingle2}
h_{\kappa, \iota} \approx  \frac{\lambda  }{4\pi d_{BU}}   e^{j\frac{-2\pi  d_{BU}  }{\lambda }} e^{j  \frac{-2\pi(\mathbf{a}_{B, \iota} - \bar{\mathbf{a}}_{U,\kappa})^T   \mathbf{e}_{BU}}{\lambda} }
\end{align}
where
\begin{subequations}
\begin{align}
d_{BU} &\triangleq \| \mathbf{p}_{B} - \mathbf{p}_{U} \|_2 \label{DisEq}\\
\mathbf{e}_{BU} &\triangleq \frac{\mathbf{p}_{B} - \mathbf{p}_{U}}{\|\mathbf{p}_{B} - \mathbf{p}_{U}\|_2} \label{DirectionEq}
\end{align}
\end{subequations}
Using spherical coordinates, the direction vector $\mathbf{e}_{BU}$ can be represented  as
\begin{align}
\mathbf{e}_{BU} = [\cos\phi\cos\theta,\; \cos\phi\sin\theta,\; \sin\phi]^T \label{ebu}
\end{align}
where $\phi$ is elevation angle and $\theta$ is azimuth angle of the LoS path.

Owing to the poor scattering property of mmWave frequency band \cite{ChannelTracking, azari2017ultra}, LoS propagation dominates  mmWave communications, and it is especially true for the link between UAV and BS. Following \cite{xu2020multiuser, xu2018robust, dabiri2020analytical}, we assume that the channel of UAV-to-BS mmWave communications is LoS-channel. Thus, the channel matrix $\mathbf{H}$ between BS and UAV is written as
\begin{align}
\mathbf{H} =& \left(
                      \begin{array}{cccc}
                        h_{1,1} &  h_{1,2} & \cdots &  h_{1,N_B} \\
                        h_{2,1} & h_{2,2} & \cdots & h_{2,N_B} \\
                        \vdots & \vdots & \ddots & \vdots \\
                        h_{N_U,1} & h_{N_U,2} & \cdots & h_{N_U,N_B} \\
                      \end{array}
 \right) \notag \\
\approx &  \frac{\lambda  }{4\pi d_{BU}}   e^{j\frac{-2\pi  d_{BU}  }{\lambda }}\left(
                      \begin{array}{cccc}
                        v_{1,1} &  v_{1,2} & \cdots &  v_{1,N_B} \\
                        v_{2,1} & v_{2,2} & \cdots & v_{2,N_B} \\
                        \vdots & \vdots & \ddots & \vdots \\
                        v_{N_U,1} & v_{N_U,2} & \cdots & v_{N_U,N_B} \\
                      \end{array}
 \right) \label{Hmatrix}
\end{align}
where $ v_{\kappa, \iota} = e^{j  \frac{-2\pi(  \mathbf{a}_{B, \iota} - \bar{\mathbf{a}}_{U,\kappa} )^T   \mathbf{e}_{BU}}{\lambda} }$. In accordance to the geometry dependent channel model in [1], [28],  Eq. \eqref{Hmatrix} can be written as
\begin{align}
\mathbf{H} \approx &  \frac{\lambda  }{4\pi d_{BU}}   e^{j\frac{-2\pi  d_{BU}  }{\lambda }} \mathbf{v}_U \mathbf{v}_B^H \label{HmatrixApprox}
\end{align}
where the array response vectors $\mathbf{v}_B$ and $\mathbf{v}_U$ are represented as
\begin{subequations}
\begin{align}
\mathbf{v}_B  = \left[e^{j  \frac{2\pi\mathbf{a}_{B, 1}^T   \mathbf{e}_{BU}}{\lambda} }, e^{j  \frac{2\pi\mathbf{a}_{B, 2}^T   \mathbf{e}_{BU}}{\lambda} },\cdots, e^{j  \frac{2\pi\mathbf{a}_{B, N_B}^T   \mathbf{e}_{BU}}{\lambda} } \right]^T   \label{Steering1}\\
\mathbf{v}_U  = \left[e^{j  \frac{2\pi \bar{\mathbf{a}}_{U, 1}^T   \mathbf{e}_{BU}}{\lambda} }, e^{j  \frac{2\pi \bar{\mathbf{a}}_{U, 2}^T   \mathbf{e}_{BU}}{\lambda} },\cdots, e^{j  \frac{2\pi \bar{\mathbf{a}}_{U, N_U}^T   \mathbf{e}_{BU}}{\lambda} } \right]^T \label{Steering2}
\end{align}
\end{subequations}

To parameterize the array response vectors, we decompose them as follows. \\
(1) \underline{Decomposition of the array response vector $\mathbf{v}_B $ at BS side.}
\vspace{0.1cm}

In \eqref{Steering1}, $\mathbf{a}_{B, \iota}^T \mathbf{e}_{BU}$ can be represented as
\begin{align}
\mathbf{a}_{B, \iota}^T \mathbf{e}_{BU} & = \frac{\lambda}{2}\left( x_\iota[1, 0, 0] \mathbf{e}_{BU}   + z_\iota [0, 0, 1] \mathbf{e}_{BU}   \right) \notag \\
& =  \frac{\lambda}{2}\left( x_\iota \Psi_B   + z_\iota  \Omega_B   \right)
\end{align}
where the two-dimensional cosine AoA/AoD of UPA at BS side is
\begin{subequations}
\begin{align}
\Psi_B & \triangleq [1, 0, 0] \mathbf{e}_{BU} = \mathbf{e}_{BU}(1)  = \cos\phi\cos\theta   \\
\Omega_B & \triangleq   [0, 0, 1] \mathbf{e}_{BU} = \mathbf{e}_{BU}(3) = \sin\phi
\end{align} \label{BSside}
\end{subequations}
\noindent Note that $\psi_B  \triangleq \cos^{-1}(\Psi_B)$ represents the AoA/AoD between $x$ axis of UPA at BS side (namely $[1, 0, 0]^T$) and the direction from UAV to BS (namely $\mathbf{e}_{BU}$), and  $\omega_B \triangleq \cos^{-1}(\Omega_B)$  represents the AoA/AoD between $z$ axis of UPA at BS side (namely $[0, 0, 1]^T$ ) and the direction from UAV to BS (namely $\mathbf{e}_{BU}$). To simplify expression, we directly name $(\Psi_B, \Omega_B)$ as two-dimensional AoA/AoD by abuse of notation hereinafter.

With $(\Psi_B, \Omega_B)$, the array response vector $\mathbf{v}_B $ can be further decomposed as
\begin{align}
\mathbf{v}_B = \mathbf{v}(\Psi_B, N_{B,x}) \otimes  \mathbf{v}(\Omega_B, N_{B,z}) \label{SteeringVB}
\end{align}
where
\begin{subequations}
\begin{align}
\mathbf{v}(\Psi_B, N_{B,x}) &= \left[1,\; e^{j  \pi \Psi_B},\; \cdots,\; e^{j  (N_{B,x} - 1)\pi \Psi_B} \right]^T \\
 \mathbf{v}(\Omega_B, N_{B,z}) & = \left[1,\; e^{j  \pi \Omega_B},\; \cdots,\; e^{j  (N_{B,z} - 1)\pi \Omega_B} \right]^T
\end{align}
\end{subequations}

\noindent (2) \underline{Decomposition of the array response vector $\mathbf{v}_U $ at UAV side.}
\vspace{0.1cm}

In \eqref{Steering2},  $\bar{\mathbf{a}}_{U, \kappa}^T   \mathbf{e}_{BU}$ can be represented as
\begin{align}
& \bar{\mathbf{a}}_{U, \kappa}^T   \mathbf{e}_{BU} =   \mathbf{a}_{U, \kappa}^T \mathbf{R}^T \mathbf{e}_{BU} \notag \\
 = &\frac{\lambda}{2}\left[x_\kappa-\frac{N_{U,x}-1}{2}, y_\kappa - \frac{N_{U,y}-1}{2}, 0\right]  \mathbf{R}^T \mathbf{e}_{BU} \notag \\
= &\frac{\lambda}{2}\left( x_\kappa[1, 0, 0] \mathbf{R}^T\mathbf{e}_{BU}   + y_\kappa [0, 1, 0] \mathbf{R}^T\mathbf{e}_{BU}   \right) \notag \\
& \qquad  -  \frac{\lambda}{2} \left[\frac{N_{U,x}-1}{2}, \frac{N_{U,y}-1}{2}, 0\right] \mathbf{R}^T\mathbf{e}_{BU}  \label{EqUAV}
\end{align}
Let
\begin{subequations}
\begin{align}  \label{psivarphi}
\Psi_U & \triangleq   [1, 0, 0] \mathbf{R}^T  \mathbf{e}_{BU}  \\
\Omega_U & \triangleq    [0, 1, 0] \mathbf{R}^T \mathbf{e}_{BU}
\end{align}
\end{subequations}
Then, \eqref{EqUAV} is rewritten as
\begin{align}
 & \quad \bar{\mathbf{a}}_{U, \kappa}^T   \mathbf{e}_{BU}  \notag \\
 = & \frac{\lambda}{2}(x_\kappa \Psi_U + y_\kappa \Omega_U) -  \frac{\lambda}{2}\left(\frac{N_{U,x}-1}{2} \Psi_U +  \frac{N_{U,y}-1}{2}\Omega_U\right)
\end{align}
Therefore, the array response vector $\mathbf{v}_U $ can be decomposed as
\begin{align}
&\qquad \mathbf{v}_U = \mathbf{v}(\Psi_U, N_{U,x}) \otimes  \mathbf{v}(\Omega_U, N_{U,y}) \label{SteeringVU}
\end{align}
where
\begin{subequations}
\begin{align}
&\quad  \mathbf{v}(\Psi_U, N_{U,x}) \notag \\
=& e^{-j \left(\frac{N_{U,x}-1}{2}\pi \Psi_U \right) }   \left[1,\; e^{j  \pi \Psi_U},\; \cdots,\; e^{j  (N_{U,x} - 1)\pi \Psi_U} \right]^T \label{UavULA1}\\
 &\quad  \mathbf{v}(\Omega_U, N_{U,y})  \notag \\
= &  e^{-j \left(\frac{N_{U,y}-1}{2}\pi \Omega_U \right) }  \left[1,\; e^{j  \pi \Omega_U},\; \cdots,\; e^{j  (N_{U,y} - 1)\pi \Omega_U} \right]^T \label{UavULA2}
\end{align}
\end{subequations}

\begin{theorem}
The approximation of the channel response in \eqref{HmatrixApprox} is written as
\begin{align} \label{CHresponse}
\mathbf{H} \approx &  \frac{\lambda}{4\pi d_{BU}}   e^{j\frac{-2\pi  d_{BU}  }{\lambda }} \Big(  \mathbf{v}(\Psi_U, N_{U,x}) \otimes  \mathbf{v}(\Omega_U, N_{U,y}) \Big)  \notag \\
 & \cdot \Big( \mathbf{v}(\Psi_B, N_{B,x})\otimes  \mathbf{v}(\Omega_B, N_{B,z})  \Big)^H
\end{align}
where the expressions of  $\Psi_B, \Omega_B$  are given by \eqref{BSside} and the expressions of $\Psi_U, \Omega_U$ are obtained as follows by
substituting \eqref{ebu} and \eqref{RotMatrix} into \eqref{psivarphi}
\begin{subequations}
\begin{align}
& \Psi_U =   \cos \alpha \cos \beta  \cos\phi\cos\theta + \sin \alpha \cos\beta  \cos\phi\sin\theta \notag \\
& \quad \quad - \sin \beta \sin\phi   \\
& \Omega_U =  (\cos \alpha \sin \beta \sin \gamma - \sin\alpha \cos\gamma) \cos\phi\cos\theta \notag \\
&  +(\sin\alpha\sin\beta\sin\gamma + \cos\alpha\cos\gamma )\cos\phi\sin\theta  +  \cos\beta \sin \gamma \sin\phi
\end{align} \label{AngleUAV}
\end{subequations}
\noindent
\end{theorem}
From \eqref{BSside} and \eqref{AngleUAV}, it can be seen that $(\Psi_B, \Omega_B)$ is only determined by UAV position, and $(\Psi_U, \Omega_U)$ is determined by both UAV position and UAV attitude. For illustrative purposes, the dependency of the two-dimensional AoA/AoD on UAV position and UAV attitude is shown in Table. \ref{angles}.

The derived channel response in \eqref{CHresponse} is in accordance with the widely used far-field mmWave channel model for UPA \cite{tsai2018millimeter,wang2017antenna}, which is characterized by the two dimensional AoA/AoD $(\Psi_B, \Omega_B)$ at BS side and $(\Psi_U, \Omega_U)$ at UAV side. A notable difference from the existing work, which simply treat the two dimensional  AoA/AoD as unknown variables to be estimated, is that \eqref{AngleUAV} extracts the relationship between the two dimensional AoA/AoD $(\Psi_B, \Omega_B)$ (or $(\Psi_U, \Omega_U)$) and UAV position \& attitude related angles.

\begin{table}
 \centering
 \begin{tabular}{|c|c|c|}
   \hline \hline
  \makecell[c]{Two-dimensional \\ AoA/AoD} &\makecell[c]{UAV position related \\ angles (i.e., $\phi$, $\theta$)    } & \makecell[c]{UAV attitude related \\ angles (i.e., $\alpha$, $\beta$, $\gamma$ )    } \\
  \hline
  ($\Psi_B, \Omega_B$)  &\makecell[c]{Dependent  }  & \makecell[c]{Independent} \\
  \hline
  ($\Psi_U, \Omega_U$) &\makecell[c]{Dependent  }  & \makecell[c]{Dependent} \\
   \hline \hline
 \end{tabular}
 \caption{Dependency of the two-dimensional AoA/AoD (BS side and UAV side) on UAV position and UAV attitude} \label{angles}
\end{table}

\section{UAV Jittering Effects Analysis}

In this section, we investigate the effects of UAV jitter on mmWave channel response.  As indicated by Theorem 1 that the jittering effects merely work on the channel parameters at UAV side, we will thus primarily focus on the analysis of the two-dimensional AoA/AoD $\Psi_U, \Omega_U$.

\subsection{The physical meaning of $\Psi_U, \Omega_U$}
To obtain  insights into jittering effects, we firstly explore the physical meaning of the two dimensional  AoA/AoD $\Psi_U, \Omega_U$.

The UPA can be decomposed into a set of vertical/horizontal uniform linear arrays (ULAs). For the purpose of illustration, the spatial relationship of the vertical/horizontal ULA and incident/emergent electromagnetic (EM) wave is shown in \figref{UPAangle1}. When UAV is in the plane $z=0$ and facing towards positive x axis (left side of \figref{UPAangle1}),  the direction vector of the vertical ULA is $\mathbf{e}_{U, v} \triangleq [1,0,0]^T$, and the direction vector of the horizontal ULA is $\mathbf{e}_{U,h}\triangleq [0,1,0]^T$. While, the direction vector of the EM wave between UAV and BS is $\mathbf{e}_{BU}$. Thus, the dot product of  $\mathbf{e}_{U, v}$ and $\mathbf{e}_{BU}$ is
\begin{align}
\mathbf{e}_{U, v}  \cdot \mathbf{e}_{BU}  = \mathbf{e}_{U, v}^T \mathbf{e}_{BU}  =  \cos \psi_U =  \Psi_U
\end{align}
where $\psi_U$ is the angle between $\mathbf{e}_{U, v} $ and  $\mathbf{e}_{BU}$,   and the dot product of  $\mathbf{e}_{U, H}$ and $\mathbf{e}_{BU}$ is
\begin{align}
\mathbf{e}_{U, h}  \cdot \mathbf{e}_{BU}  = \mathbf{e}_{U, h}^T \mathbf{e}_{BU}  =  \cos \omega_U =  \Omega_U
\end{align}
where $\omega_U $ is the angle between $\mathbf{e}_{U, h}$ and  $\mathbf{e}_{BU}$.

When the UPA is rotated by $\mathbf{R}$ (right side of \figref{UPAangle1}), the direction of the vertical ULA turns to be $\mathbf{R} \mathbf{e}_{U, v} $, and the direction of the horizontal ULA turns to be $\mathbf{R} \mathbf{e}_{U,h} $, thus the dot products of the direction vectors of ULAs and EM wave become
\begin{subequations}
\begin{align}
\mathbf{R}\mathbf{e}_{U, v}  \cdot \mathbf{e}_{BU}  = \mathbf{e}_{U, v}^T \mathbf{R}^T\mathbf{e}_{BU}  =  \cos \psi_U =  \Psi_U   \\
\mathbf{R}\mathbf{e}_{U, h}  \cdot \mathbf{e}_{BU}  = \mathbf{e}_{U, h}^T \mathbf{R}^T\mathbf{e}_{BU}  =  \cos \omega_U =  \Omega_U
\end{align}
\end{subequations}
To summarize, $\Psi_U$ and $\Omega_U$ are the cosine of the angles between the decomposed ULAs (vertical/horizontal) and incident/emergent EM wave.

\begin{figure}[tp]{
\begin{center}{\includegraphics[width=9cm ]{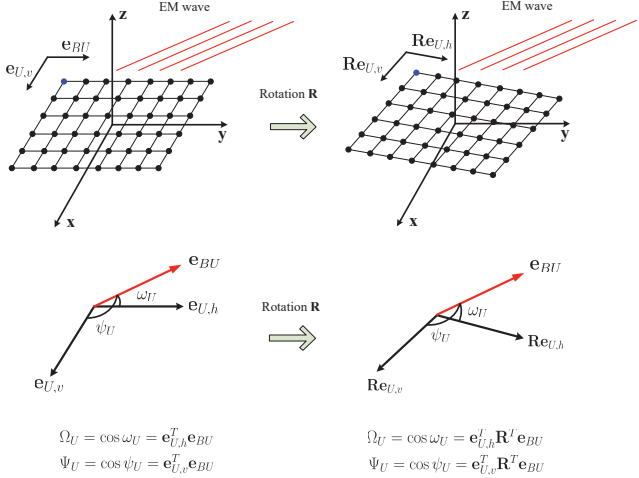}}
\caption{Geometric interpretation of the two-dimensional AoA/AoD at UAV side}\label{UPAangle1}
\end{center}}
\end{figure}

\subsection{UAV Jittering Effects on $\Psi_U, \Omega_U$}

The most direct reflection of jittering effects on UAV is the fluctuation of attitude angles around the desired attitude \cite{fuertes2019multirotor}. To investigate jittering effects on UAV mmWave channel, we should firstly model the jittering effects on UAV attitude. To this end, we represent the instantaneous attitude angles of the UAV as
\begin{align}
\alpha = \bar\alpha + \Delta\alpha \notag \\
\beta = \bar\beta + \Delta\beta \notag \\
\gamma = \bar\gamma + \Delta\gamma \notag
\end{align}
where $\bar\alpha, \bar\beta, \bar\gamma$  are the desired attitude angles which are configured by the flight control system, and $\Delta\alpha, \Delta\beta,  \Delta\gamma $ are the  fluctuations about the desired attitude angles caused by jittering effects. Owing to the small-scale fluctuation of the attitude angles,   the two dimensional AoA/AoD $ \Psi_U,  \Omega_U $ can be approximated using first-order Taylor polynomial, i.e.,

\begin{align} \label{TaylorEq}
&\left[
  \begin{array}{c}
    \Psi_U(\alpha, \beta, \gamma) \\
    \Omega_U(\alpha, \beta, \gamma) \\
  \end{array}
\right] \approx \left[
  \begin{array}{c}
     {\Psi_U}(\bar\alpha, \bar\beta, \bar\gamma) \\
     {\Omega}_U(\bar\alpha, \bar\beta, \bar\gamma)  \\
  \end{array}
\right]   +    \mathbf{J}(\bar\alpha, \bar\beta, \bar\gamma)  \left[
               \begin{array}{c}
                 \Delta\alpha \\
                 \Delta\beta \\
                 \Delta\gamma \\
               \end{array}
             \right]
\end{align}
where the Jacobian  matrix $\mathbf{J}(\alpha, \beta, \gamma)$ is represented as
\begin{align}
\mathbf{J}(\alpha, \beta, \gamma) = \left[
                                             \begin{array}{ccc}
                                               \frac{\partial \Psi_U}{\partial \alpha} &  \frac{\partial \Psi_U}{\partial \beta}  &  \frac{\partial \Psi_U}{\partial \gamma} \\
                                                \frac{\partial \Omega_U}{\partial \alpha}  & \frac{\partial \Omega_U}{\partial \beta} & \frac{\partial \Omega_U}{\partial \gamma}  \\
                                             \end{array}
\right]
\end{align}
and
\begin{align}
\begin{split}
 \frac{\partial \Psi_U}{\partial \alpha} &=  -\sin \alpha \cos \beta  \cos\phi\cos\theta + \cos \alpha \cos\beta  \cos\phi\sin\theta  \notag \\
 \frac{\partial \Psi_U}{\partial \beta} &=  -\cos \alpha \sin \beta  \cos\phi\cos\theta - \sin \alpha \sin\beta  \cos\phi\sin\theta  \notag \\
 &- \cos \beta \sin\phi \notag \\
 \frac{\partial \Psi_U}{\partial \gamma} &= 0 \notag \\
 \frac{\partial \Omega_U}{\partial \alpha} &= (-\sin \alpha \sin \beta \sin \gamma - \cos\alpha \cos\gamma) \cos\phi\cos\theta \notag \\
&  +(\cos\alpha\sin\beta\sin\gamma - \sin\alpha\cos\gamma )\cos\phi\sin\theta  \notag \\
 \frac{\partial \Omega_U}{\partial \beta} &=   (\cos \alpha \cos \beta \sin \gamma ) \cos\phi\cos\theta \notag \\
&  +(\sin\alpha\cos\beta\sin\gamma  )\cos\phi\sin\theta  -  \sin\beta \sin \gamma \sin\phi \notag \\
 \frac{\partial \Omega_U}{\partial \gamma}  &=  (\cos \alpha \sin \beta \cos \gamma + \sin\alpha \sin\gamma) \cos\phi\cos\theta \notag \\
&  +(\sin\alpha\sin\beta\cos\gamma - \cos\alpha\sin\gamma )\cos\phi\sin\theta  \notag \\
&+  \cos\beta \cos \gamma \sin\phi
\end{split}
\end{align}

Following \cite{dabiri2020analytical,yuan2020learning, kaadan2014multielement}, we model the attitude angles under jittering effects by Gaussian distribution,  i.e., $\Delta \alpha \sim \mathcal{N}(0, \sigma_\alpha^2)$, $\Delta\beta \sim \mathcal{N}(0, \sigma_\beta^2)$, and $\Delta\gamma \sim \mathcal{N}(0, \sigma_\gamma^2)$ and $\Delta\alpha, \Delta\beta, \Delta\gamma$  are independent of each other. Then, the following theorem is obtained about the distribution of $ [ {\Psi_U}(\alpha, \beta, \gamma), {\Omega}_U(\alpha, \beta, \gamma)]^T $ under jittering effects.

\begin{theorem}
The two dimensional AoA/AoD at UAV side, i.e., $ [ {\Psi_U}(\alpha, \beta, \gamma), {\Omega}_U(\alpha, \beta, \gamma)]^T $,  follows a Gaussian distribution, i.e., $ [ {\Psi_U}(\alpha, \beta, \gamma), {\Omega}_U(\alpha, \beta, \gamma)]^T \sim \mathcal{N}(\boldsymbol{\mu}, \boldsymbol{\Sigma}) $,
whose mean vector is
\begin{align}
\boldsymbol{\mu} =  \left[\begin{array}{c}
     {\Psi_U}(\bar\alpha, \bar\beta, \bar\gamma) \\
     {\Omega}_U(\bar\alpha, \bar\beta, \bar\gamma)  \\
  \end{array}\right]
\end{align}
and covariance matrix is
\begin{align}
\boldsymbol{\Sigma}=  \mathbf{J}(\bar\alpha, \bar\beta, \bar\gamma)  \left[
                                                                          \begin{array}{ccc}
                                                                            \sigma_\alpha^2 & 0 & 0 \\
                                                                            0 & \sigma_\beta^2 & 0 \\
                                                                            0 & 0 & \sigma_\gamma^2 \\
                                                                          \end{array}
                                                                        \right]
  \mathbf{J}(\bar\alpha, \bar\beta, \bar\gamma)^T
\end{align}
\end{theorem}

\begin{proof}
As $\Delta \alpha, \Delta \beta, \Delta \gamma$ follow Gaussian distribution, according to \eqref{TaylorEq}, it is easy to find that $[{\Psi_U}(\alpha, \beta, \gamma), {\Omega}_U(\alpha, \beta, \gamma)]^T$ also follows a Gaussian distribution. The mean vector and covariance matrix of $[{\Psi_U}(\alpha, \beta, \gamma), {\Omega}_U(\alpha, \beta, \gamma)]^T$ can be obtained according to their definitions, and the derivations are omitted as they both follow standard procedures.
\end{proof}

\begin{remark}
According to Theorem 2,  the mean vector is merely determined by the desired attitude $(\bar\alpha, \bar\beta, \bar\gamma)$, but the covariance matrix is jointly determined by  the scale of fluctuation of attitude angles, i.e., $(\sigma_\alpha^2, \sigma_\beta^2, \sigma_\gamma^2)$, and the desired attitude $(\bar\alpha, \bar\beta, \bar\gamma)$. In other words, even with the same $(\sigma_\alpha^2, \sigma_\beta^2, \sigma_\gamma^2)$, the jittering effects on UAV mmWave channel response might not be the same.
\end{remark}

According to the properties of the marginal distribution of a bivariate Gaussian distribution \cite{tong2012multivariate}, the marginal distributions of ${\Psi_U}(\alpha, \beta, \gamma)$ and ${\Omega_U}(\alpha, \beta, \gamma)$ are obtained as follows.
\begin{lemma}
The marginal distribution of ${\Psi_U}(\alpha, \beta, \gamma)$ is
\begin{align}
{\Psi_U}(\alpha, \beta, \gamma) \sim \mathcal{N}\left({\Psi_U}(\bar\alpha, \bar\beta, \bar\gamma), \boldsymbol{\Sigma}(1,1)\right)
\end{align}
where $\boldsymbol{\Sigma}(1,1) =    \sigma_\alpha^2 \left(\frac{\partial \Psi_U}{\partial \alpha} \right)^2  + \sigma_\beta^2 \left(\frac{\partial \Psi_U}{\partial \beta}\right)^2  + \sigma_\gamma^2 \left(\frac{\partial \Psi_U}{\partial \gamma}\right)^2$,
and the marginal distribution of ${\Omega_U}(\alpha, \beta, \gamma)$ is
\begin{align}
{\Omega_U}(\alpha, \beta, \gamma) \sim \mathcal{N}\left({\Omega_U}(\bar\alpha, \bar\beta, \bar\gamma), \boldsymbol{\Sigma}(2,2)\right)
\end{align}
where $\boldsymbol{\Sigma}(2,2) =  \sigma_\alpha^2 \left(\frac{\partial \Omega_U}{\partial \alpha} \right)^2  + \sigma_\beta^2 \left(\frac{\partial \Omega_U}{\partial \beta}\right)^2  + \sigma_\gamma^2 \left(\frac{\partial \Omega_U}{\partial \gamma}\right)^2$.
\end{lemma}

\begin{figure}[t]
\begin{minipage}[!h]{0.48\linewidth}
\centering
\includegraphics[ width=1.1\textwidth]{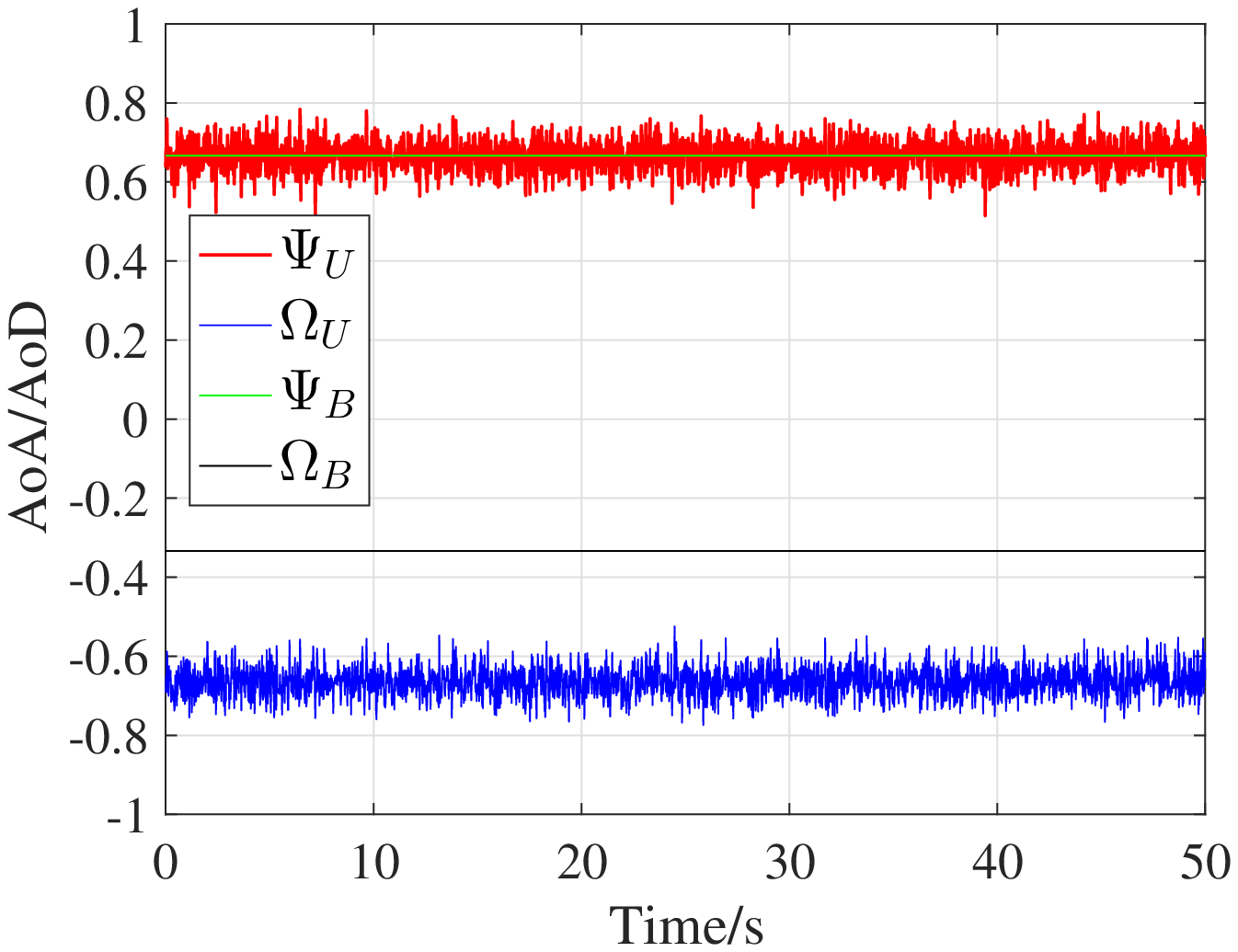}
\subcaption{\scriptsize Variation of $\Psi_U$, $\Omega_U$, $\Psi_B$ and $\Omega_B$  over time (Scenario 1)}
\label{Scenario1Time}
\end{minipage}
\begin{minipage}[!h]{0.48\linewidth}
\centering
\includegraphics[ width=1.1\textwidth]{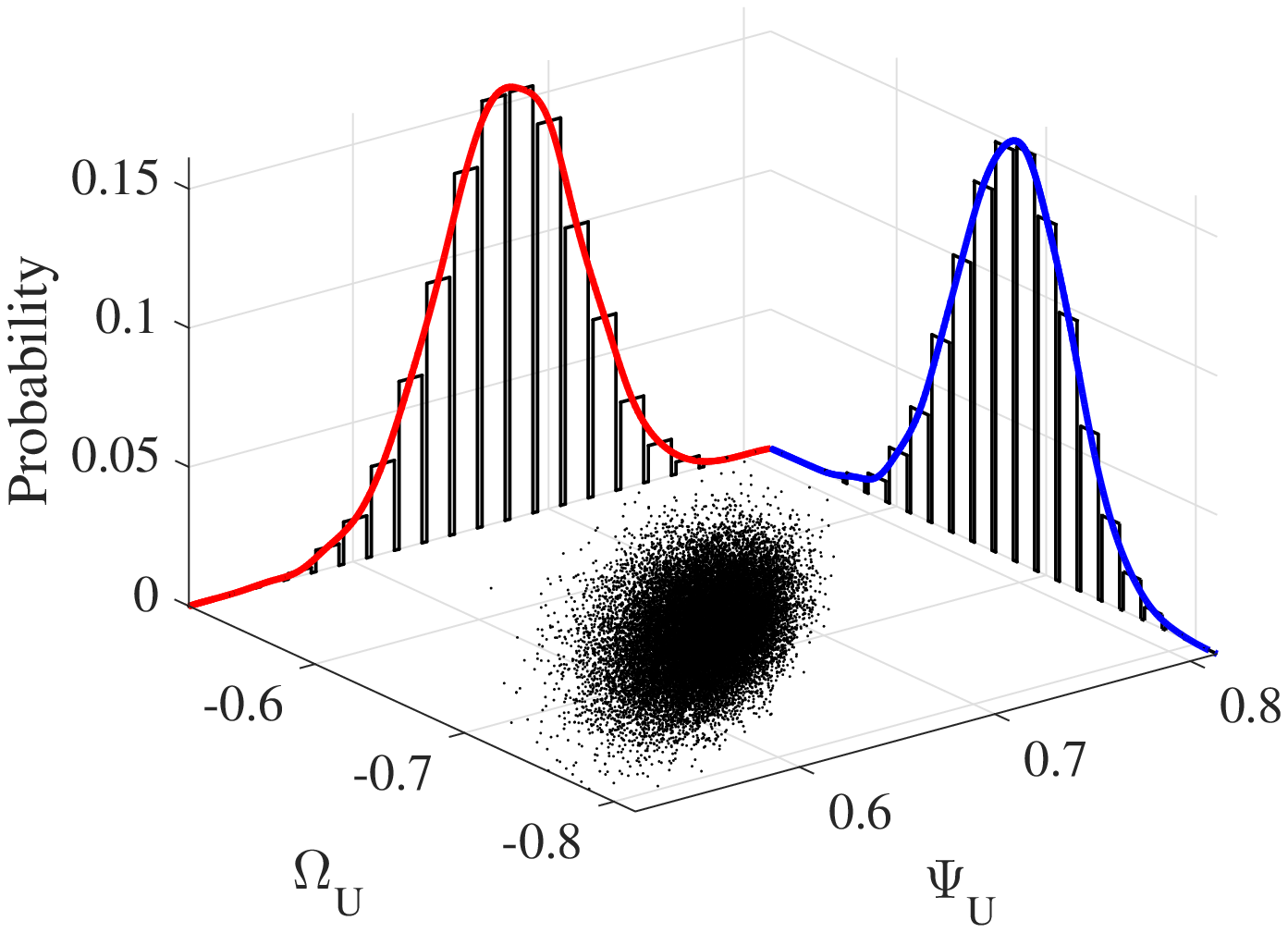}
\subcaption{\scriptsize Empirical marginal probability density function of $\Psi_U$ and $\Omega_U$ (Scenario 1)}
\label{Scenario1PDF}
\end{minipage}
\begin{minipage}[!h]{0.48\linewidth}
\centering
\includegraphics[ width=1.1\textwidth]{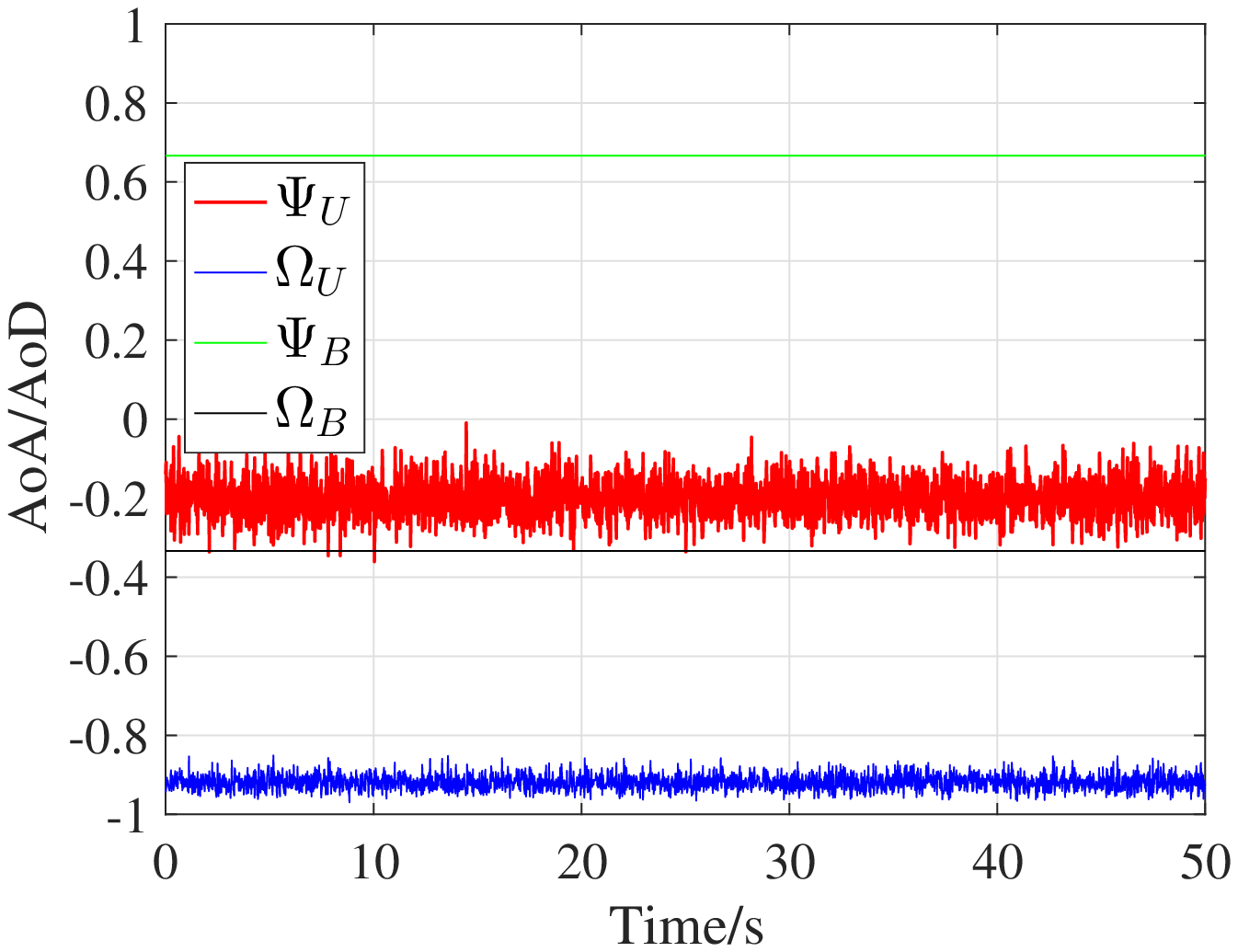}
\subcaption{\scriptsize Variation of $\Psi_U$, $\Omega_U$, $\Psi_B$ and $\Omega_B$  over time (Scenario 2)}
\label{Scenario2Time}
\end{minipage}
\hspace{0.17cm}\begin{minipage}[!h]{0.48\linewidth}
\centering
\includegraphics[ width=1.1\textwidth]{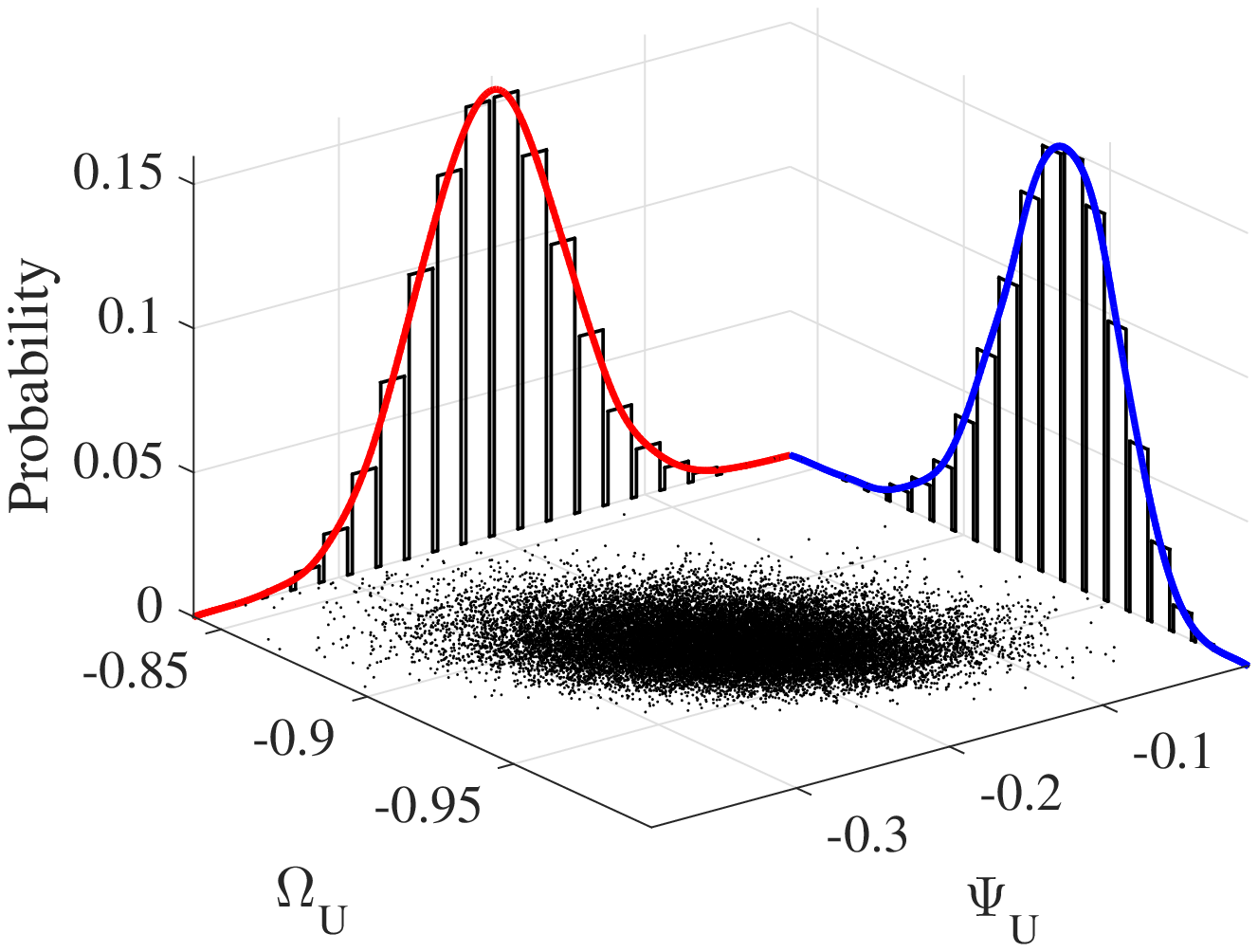}
\subcaption{\scriptsize Empirical marginal probability density function of $\Psi_U$ and $\Omega_U$ (Scenario 2)}
\label{Scenario2PDF}
\end{minipage}
\begin{minipage}[!h]{0.48\linewidth}
\centering
\includegraphics[ width=1.1\textwidth]{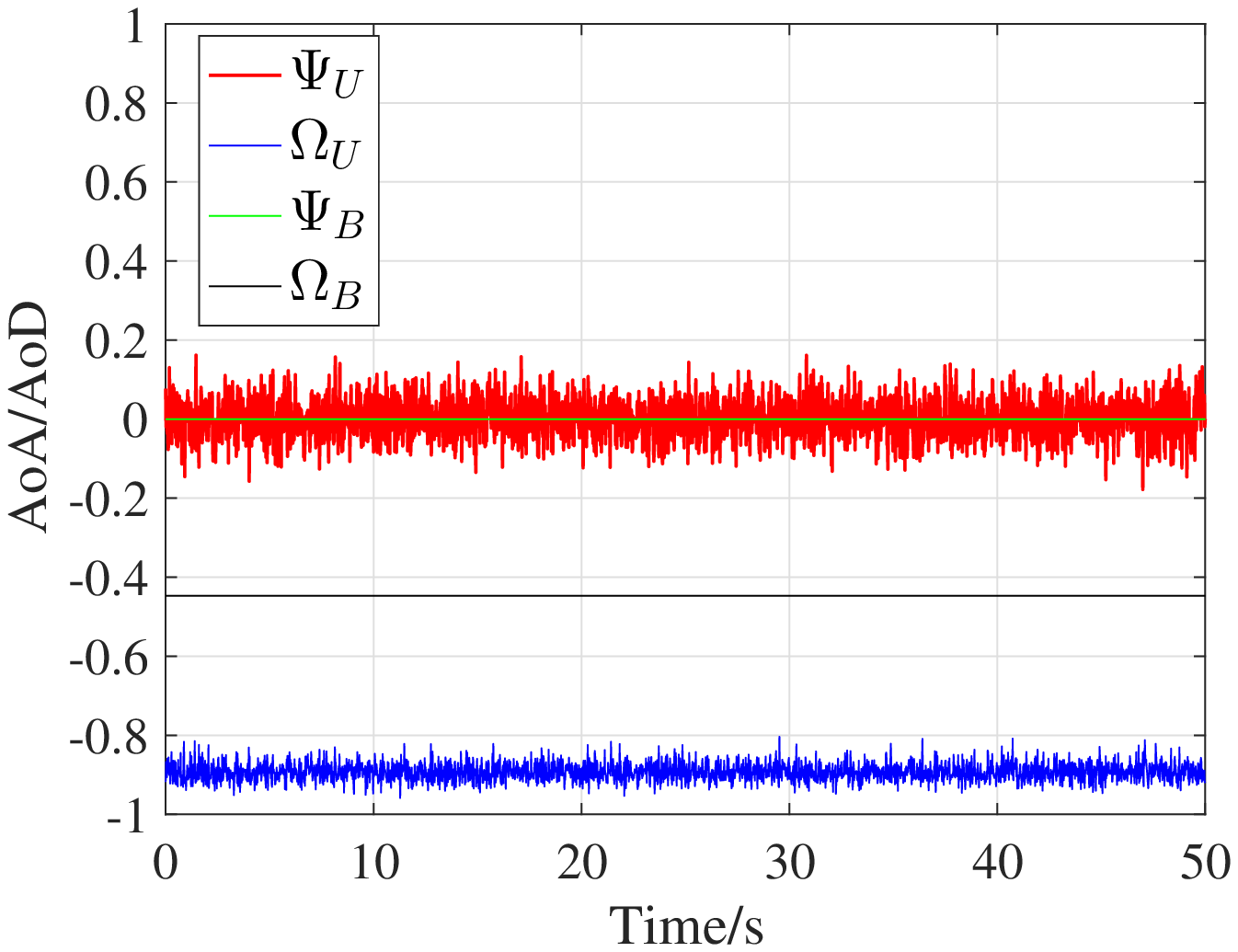}
\subcaption{\scriptsize Variation of $\Psi_U$, $\Omega_U$, $\Psi_B$ and $\Omega_B$   over time (Scenario 3)}
\label{Scenario3Time}
\end{minipage}
\hspace{0.17cm}\begin{minipage}[!h]{0.48\linewidth}
\centering
\includegraphics[ width=1.1\textwidth]{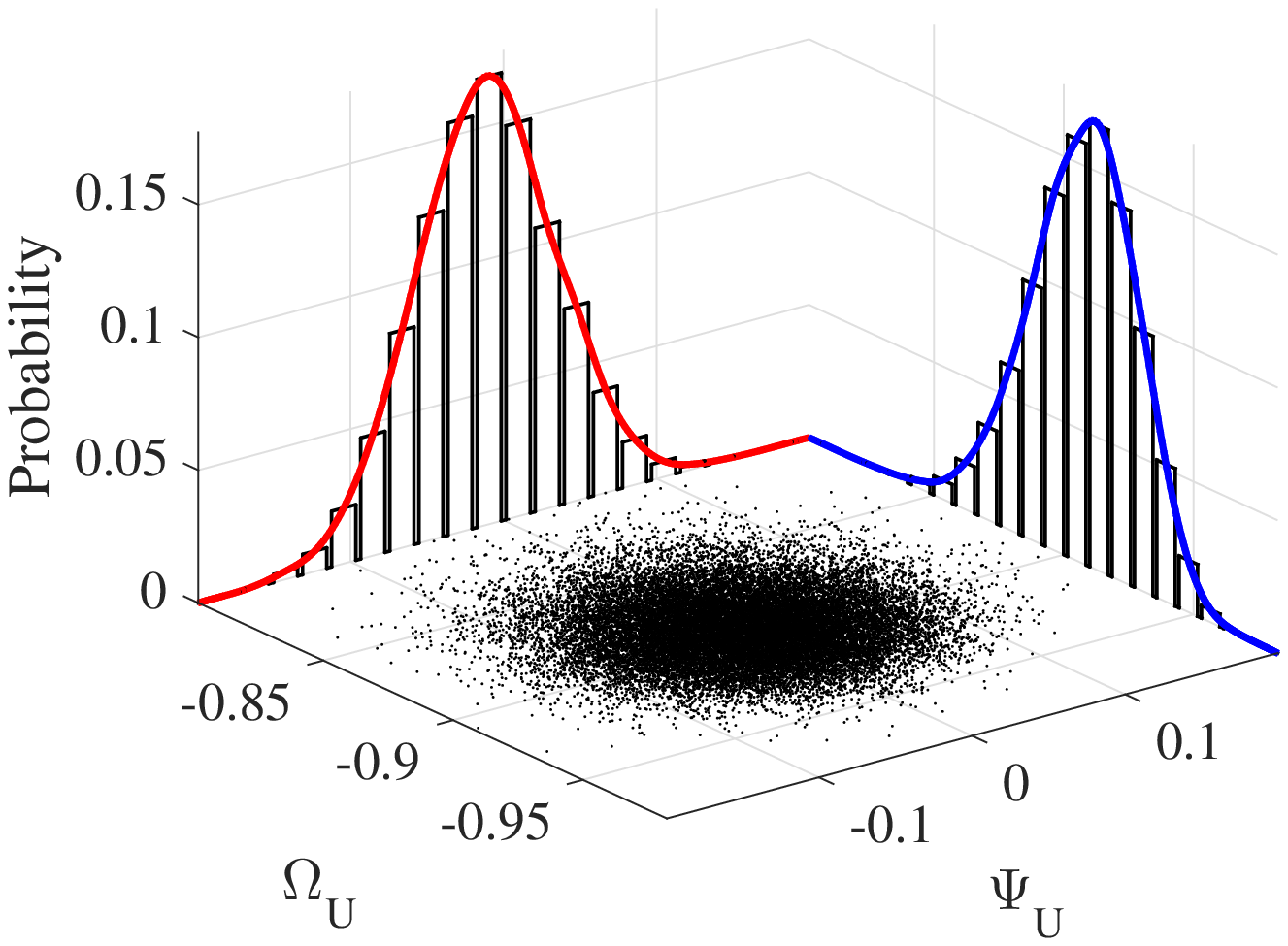}
\subcaption{\scriptsize Empirical marginal probability density function of $\Psi_U$ and $\Omega_U$ (Scenario 3)}
\label{Scenario3PDF}
\end{minipage}
\caption{Jittering effects on two-dimensional AoA/AoD with different UAV positions and UAV attitudes} \label{AoAFluctuation}
\end{figure}

\subsection{Numerical Examples}

To provide further insights into jittering effects on UAV mmWave channel, we carry out numerical simulations in the following three scenarios. \\
\underline{\emph{Scenario 1.}} UAV is hovering at the position $\mathbf{p}_U = [-100, 100, 50]^T$ (Cartesian coordinates) with its desired flight attitude being $\bar\alpha = \bar\beta = \bar\gamma = 0$ rad; \\
\underline{\emph{Scenario 2.}}  UAV is hovering at the position $\mathbf{p}_U = [-100, 100, 50]^T$ with its desired flight attitude being $\bar\alpha = 1$ rad, $\bar\beta = \bar\gamma = 0$ rad; \\
\underline{\emph{Scenario 3.}} UAV is hovering at the position  $\mathbf{p}_U = [0, 100, 50]^T$ with its desired flight attitude being $\bar\alpha = \bar\beta = \bar\gamma = 0$ rad; \\

In all the three scenarios, we set $N_{B,x} = N_{B,z} =  N_{U,x} = N_{U,y} = 16$ and the operating frequency as $28$GHz. We assume that the angles of yaw, pitch and roll follow Gaussian distribution with the standard deviation $\sigma_\alpha = \sigma_\beta = \sigma_\gamma = 0.05$, which means $99.73\%$ of the values of $\alpha$, $\beta$, $\gamma$ are less than $0.15$ rad (equivalently $8.59$ degree) biased from their expected values.

From \figref{Scenario1Time}, \figref{Scenario2Time}, and \figref{Scenario3Time}, we can find that UAV jitter does not cause fluctuation to $\Omega_B$ and $\Psi_B$. With respect to $\Omega_U$ and $\Psi_U$, we firstly derive the  mean vector and covariance matrix according to Theorem 2 and then based on which we will study the three examples respectively.

In Scenario 1, the mean vector and covariance matrix are
\begin{align}
\boldsymbol{\mu} &= [0.6667, -0.6667]^T , \;
\boldsymbol{\Sigma} &=
 \left[
             \begin{array}{cc}
               0.0014 & 0.0011 \\
               0.0011 & 0.0014   \\
             \end{array}
           \right] \notag
\end{align}
According to Lemma 1, the standard deviation of the marginal distribution of $\Psi_U$ is $0.0374$ and the standard deviation of the marginal distribution $\Omega_U$ is $0.0374$. Thus, the $99.73\%$ confidence interval (three-sigma interval) of $\Psi_U$ is $(0.5545, 0.7789)$, and the $99.73\%$ confidence interval of $\Omega_U$ is $(-0.7789, -0.5545)$. From \figref{Scenario1Time} and \figref{Scenario1PDF}, $\Omega_U$ and $\Psi_U$ both fluctuate around $0.6667$ and  $-0.6667$, respectively, and the degrees of the fluctuations are almost the same, which is in accordance to our theoretical expectation.

In Scenario 2, UAV hovers in the same position as Scenario 1 but changes its facing direction from $\bar\alpha=0$ to $\bar\alpha=1$. The mean vector and covariance matrix are
\begin{align}
\boldsymbol{\mu}  = [-0.2008, -0.9212]^T, \; \boldsymbol{\Sigma}  =
 \left[
             \begin{array}{cc}
               0.0024 & -0.0005 \\
               -0.0005 & 0.0004   \\
             \end{array}
           \right] \notag
\end{align}
It is noteworthy that the mean vector and covariance matrix are different from Scenario 1. According to Lemma 1, the standard deviation of the marginal distribution of $\Psi_U$ is $0.0489$ and the standard deviation of the marginal distribution $\Omega_U$ is $0.02$. Thus, the $99.73\%$ confidence interval of $\Psi_U$ is $(-0.3475, -0.0541)$, and the $99.73\%$ confidence interval of $\Omega_U$ is $(-0.9812, -0.8612)$. From \figref{Scenario2Time} and \figref{Scenario2PDF},  we can see that (1) the expected values of $\Omega_U$ and $\Psi_U$ are different from Scenario 1, and (2) $\Omega_U$ experiences mild fluctuation, but $\Psi_U$ experiences intense fluctuation, which validates our conclusion in Theorem 2 that UAV attitude is a determinant factor to the jittering effects on UAV mmWave channel.

In Scenario 3, UAV's attitude is the same as Scenario 1, but UAV's position changes to $\mathbf{p}_U = [0, 100, 50]^T$. The mean vector and covariance matrix are
\begin{align}
\boldsymbol{\mu}  = [0, -0.8944]^T, \; \boldsymbol{\Sigma}  =
 \left[
             \begin{array}{cc}
               0.0025 & 0 \\
               0 & 0.0005   \\
             \end{array}
           \right] \notag
\end{align}
It is noteworthy that the mean vector and covariance matrix are different from Scenario 1.  According to Lemma 1, the standard deviation of the marginal distribution of $\Psi_U$ is $0.05$ and the standard deviation of the marginal distribution $\Omega_U$ is $0.0224$. Thus, the $99.73\%$ confidence interval of $\Psi_U$ is $(-0.15, 0.15)$, and the $99.73\%$ confidence interval of $\Omega_U$ is $(-0.9616, -0.8272)$. From \figref{Scenario3Time} and \figref{Scenario3PDF},   we can see that (1) the expected values of $\Omega_U$ and $\Psi_U$ are different from Scenario 1, and (2) $\Omega_U$ experiences mild fluctuation, but $\Psi_U$ experiences intense fluctuation, which validates our conclusion in Theorem 2 that UAV position is also a determinant factor to the jittering effects on UAV mmWave channel.

To summarize, the above numerical results indicate that (1) UAV jitter does not cause fluctuations to the AoA/AoD of mmWave channel at BS side; (2) even with the same statistics of the attitude angles (i.e., yaw, pitch and roll), the jittering effects on the AoA/AoD of mmWave channel vary with UAV's position and attitude.

\section{Beam Training Design for UAV MmWave Communications With Navigation Information}

Compressed sensing (CS) technique, which firstly measures channel response using the random sensing matrix and then reconstructs channel response from the measurements, has been proved to be powerful in beam training (or channel estimation) in mmWave communications \cite{wang2020orthogonal, gao2017reliable}. In UAV communication, through the relationship between the two dimensional  AoA/AoD and UAV position \& attitude angles,  navigation information can be utilized to assist beam training.  In this section, UAV navigation information assisted CS based beam training scheme is proposed to estimate the AoA/AoD of UAV mmWave channel.

\subsection{Review of CS Based Beam Training}

In downlink transmission, the  temporal model of the received signal at UAV side is given by
\begin{align}
y_n  = \mathbf{m}_{U,n}^H \mathbf{H} \mathbf{f}_{B,n} s + \mathbf{m}_{U,n}^H \bar{\mathbf{w}}_n \label{sensing}
\end{align}
where $n$ is time index, and without loss of generality we let the transmitted pilot signal be $s=1$ in the following context. At the training stage, the transmit beamforming vector $ \mathbf{f}_{B,n}$ and the receive beamforming vector $\mathbf{m}_{U,n}$ are configured to measure the AoA/AoD of the LoS path. Substituting \eqref{CHresponse} into \eqref{sensing}, we have
\begin{align}
y_n  = \tau \cdot (\mathbf{f}^*_{B,n} \otimes \mathbf{m}_{U,n} )^H   (\mathbf{v}_B \otimes \mathbf{v}_U ) + \mathbf{m}_{U,n}^H \bar{\mathbf{w}}_n
\end{align}
where $\tau =  \frac{\lambda}{4\pi d_{BU}}   e^{j\frac{-2\pi  d_{BU}  }{\lambda }}$. By concatenating $N$ channel measurements, the vector-form received signal is given by
\begin{align}
\mathbf{y}  & = \tau \cdot \mathbf{D} (\mathbf{v}_B \otimes \mathbf{v}_U ) + \mathbf{w}
\end{align}
where
\begin{align}
\mathbf{y} & = \left[y_{1},\; y_{2},\;\cdots, y_{N} \right]^T \notag\\
\mathbf{D} & = \left[\mathbf{f}_{B,1}^* \otimes \mathbf{m}_{U,1},\; \mathbf{f}_{B,2}^*  \otimes \mathbf{m}_{U,2},\; \cdots, \;\mathbf{f}_{B,N}^*  \otimes \mathbf{m}_{U,N}\right]^H\notag \\
\mathbf{w} &= \left[\mathbf{m}_{U,1}^H\bar{\mathbf{w}}_1,\; \mathbf{m}_{U,2}^H\bar{\mathbf{w}}_2,\; \cdots, \;\mathbf{m}_{U,N}^H \bar{\mathbf{w}}_N\right]^T \notag
\end{align}
The matrix $\mathbf{D}$ is the sensing matrix. In the realm of compressed sensing, a good sensing matrix should satisfy restricted isometry property (RIP) and achieve low mutual coherence \cite{CompressedSensingBook}. Random matrices, e.g., Gaussian random matrix,  have a high probability of meeting the desired properties. Due to the constant modulus constraint of mmWave array antenna with hybrid structure, the sensing matrix can be constructed by imposing constant modulus constraint on the matrix elements and randomizing their phases \cite{gao2017reliable, wang2020orthogonal}. It is noteworthy that the  sensing range of randomly generated sensing matrix is semi-omnidirectional, i.e., the sensing range is
\begin{align}
&\Psi_B, \Omega_B \in (-1,1) \notag \\
&\Psi_U, \Omega_U \in (-1,1) \notag
\end{align}

However, a salient advantage of mmWave communication system mounted on UAV platforms is that UAV navigation system consists of GPS, barometer, and IMU, e.g., gyroscope and accelerometer, which are capable of providing estimation of 3D UAV position  and UAV attitude  \cite{DJI}. With the relationship between UAV navigation information and AoA/AoD of mmWave channel respoonse (i.e., \eqref{BSside} and \eqref{AngleUAV}), we can further narrow down the sensing range of beam training. The reduced sensing range and the more concentrated radiation power of training signal, as the benefits of navigation information, are expected to render beam training more accurate and more training cost effective. To attain the goal, we should first investigate the extent that the sensing range of AoA/AoD can be reduced by UAV navigation information and then propose a scheme for CS-based beam training to constrain the sensing range.

\begin{figure*}[t]
\begin{minipage}[!h]{0.32\linewidth}
\centering
\includegraphics[ width=1.1\textwidth]{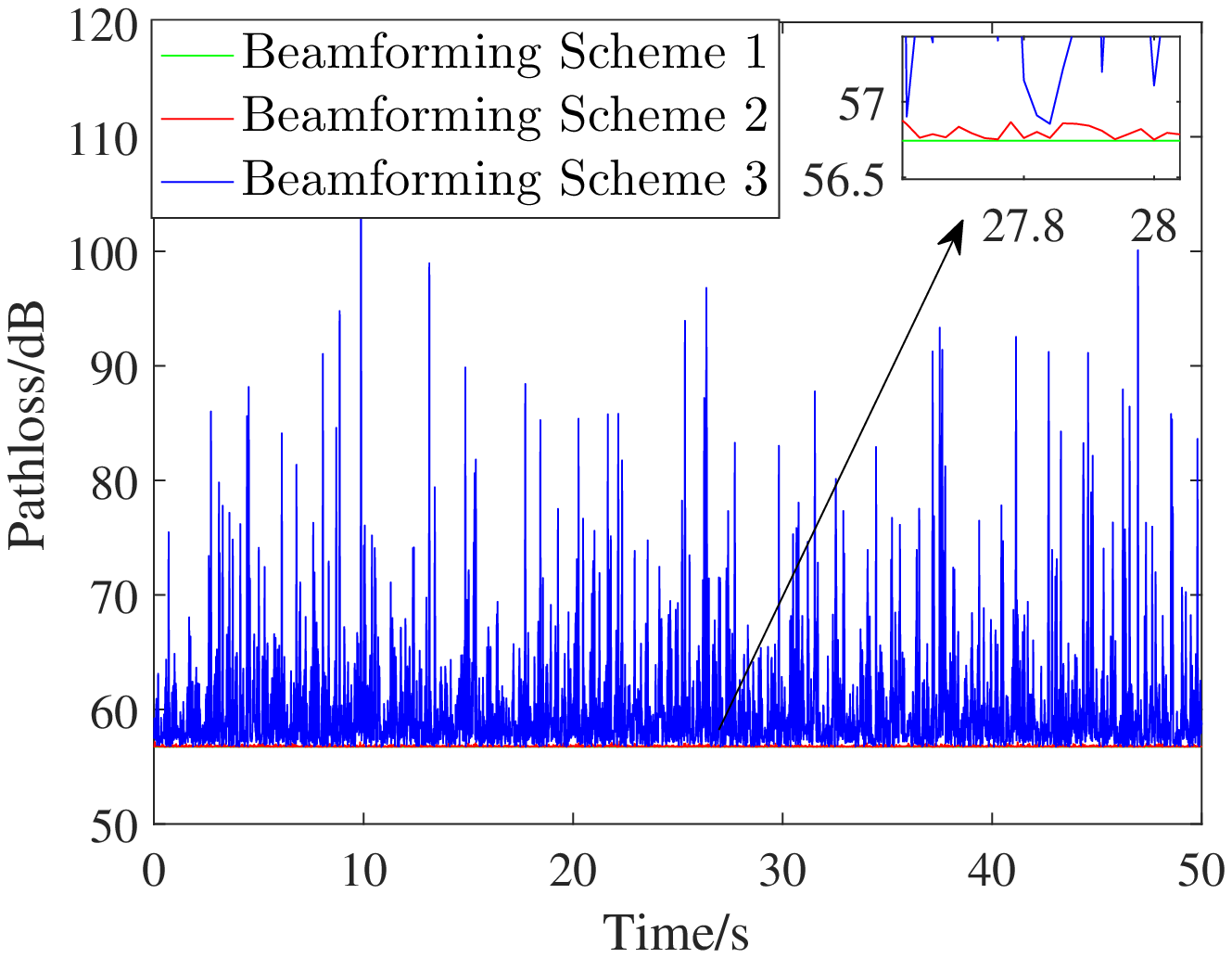}
\subcaption{\footnotesize Scenario1}
\label{Scenario1}
\end{minipage}
\begin{minipage}[!h]{0.32\linewidth}

\centering
\includegraphics[ width=1.1\textwidth]{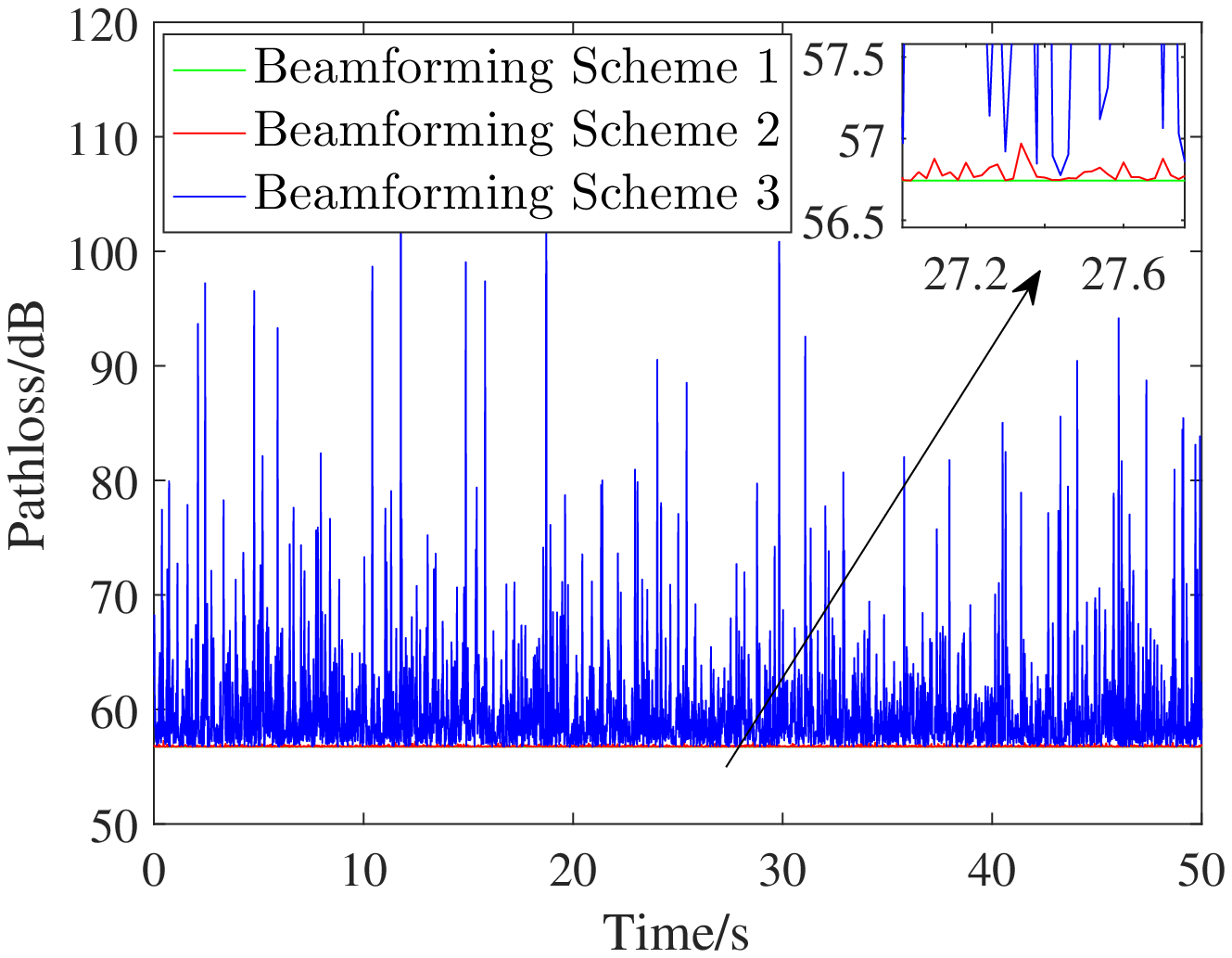}
\subcaption{\footnotesize Scenario2}
\label{Scenario2}
\end{minipage}
\begin{minipage}[!h]{0.32\linewidth}
\centering
\includegraphics[ width=1.1\textwidth]{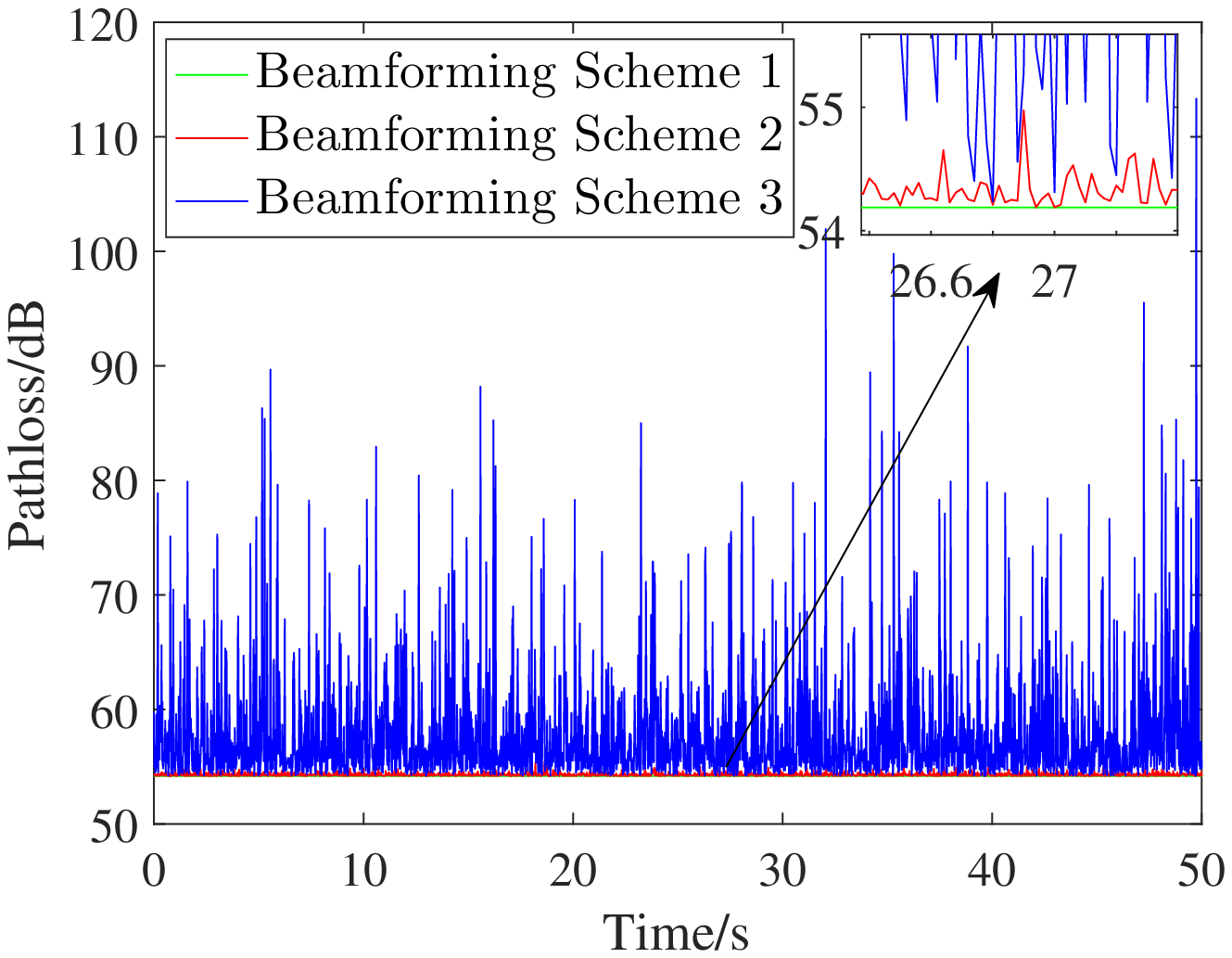}
\subcaption{\footnotesize Scenario3}
\label{Scenario3}
\end{minipage}
\caption{Variation of free space path loss over time in different scenarios} \label{Pathloss}
\end{figure*}

\subsection{The Role of UAV Navigation Information in AoA/AoD Estimation}
The accuracy of 2-D positioning of civilian GPS can  be less than $1$ meter in terms of standard deviation \cite{wing2005consumer}, while the accuracy of altitude estimation of GPS is generally low, which is approximately 5 meters. However, by fusing the information of barometer and GPS, the estimation of altitude can be improved to $\leq 0.5$ meter \cite{zaliva2014barometric}. The yielded GPS coordinates $\hat{\mathbf{p}}_U = (\hat{x}, \hat{y}, \hat{z})^T$ of UAV can be fed back to BS and utilized by both sides for AoA/AoD estimation. To facilitate analysis, we assume that
\begin{align}
\hat{\mathbf{p}}_U = (\hat{x}, \hat{y}, \hat{z})^T  = (x, y, z)^T + \mathbf{n} \label{UAVpos}
\end{align}
where $\mathbf{n}\sim \mathcal{N}(\mathbf{0}_{3\times 1}, \mathbf{I}_{3\times 3})$, which is harsher than the best accuracy that can be achieved in practice.

With regard to UAV attitude, a widely adopted solution is to fuse the information provided by gyroscope and accelerometer \cite{GyroAcc}. Through integrating the angular velocity generated by gyroscope, rotation angles can be obtained. Drift is the main drawback of gyroscope based rotation angles estimation. To compensate for the drift, accelerometer should be used. However, gyroscope and accelerometer can be very sensitive to the jitter/vibration induced by the engine and wind gust \cite{mahony2008nonlinear, weinberg2011gyro}, and the estimation error increases with the degree of jittering effects. Considering the fact that the long term attitude of the aircraft is much easier to be estimated than its instantaneous attitude, we assume that inertial units can accurately yield the expected values of rotation angles, i.e.,
\begin{align}
(\hat\alpha, \hat\beta, \hat\gamma)^T = (\bar{\alpha}, \bar{\beta}, \bar{\gamma})^T
\end{align}

According to the relationship derived in \eqref{BSside} and \eqref{AngleUAV}, the  navigation-system-yielded UAV position $(\hat{x}, \hat{y}, \hat{z})^T$ and UAV attitude $(\hat\alpha, \hat\beta, \hat\gamma)^T$   can be  transformed into the two-dimensional AoA/AoD $(\hat{\Psi}_B, \hat{\Omega}_B)$ and  $(\hat{\Psi}_U, \hat{\Omega}_U)$ as a priori knowledge for the subsequent beam training process. To figure out the degree that we can trust the a priori knowledge,  we will investigate the performance loss caused by the direct application of $(\hat{\Psi}_B, \hat{\Omega}_B)$ and  $(\hat{\Psi}_U, \hat{\Omega}_U)$  to UAV mmWave communications, and based on which we will determine whether they need further refinement by beam training.

It is widely acknowledged that analog beamforming, which shapes the beam through a single radio frequency (RF) chain for all the element antennas, is considered mandatory for mmWave communication systems, e.g., 5G NR \cite{giordani2018tutorial}. With beamforming, the UPA of UAV is indeed a directional antenna whose antenna gain is dependent on the estimation accuracy of $({\Psi}_B, {\Omega}_B)$ and  $({\Psi}_U, {\Omega}_U)$.  To study the free-space path loss of UAV mmWave communications with beamforming, we assume a downlink scenario where BS transmits and UAV receives. The received signal at UAV side is written by
\begin{align}
y  = \mathbf{m}_{U}^H \mathbf{H} \mathbf{f}_{B}s + \mathbf{m}_{U}^H \bar{\mathbf{w}}
\end{align}
where  $s$ is the transmitted signal, $\bar{\mathbf{w}} \sim \mathcal{CN}(\mathbf{0}_{N_U \times 1}, \sigma^2\mathbf{I}_{N_U \times N_U})$ is additive white Gaussian noise, $ \mathbf{m}_{U}$ is the receive beamforming vector at UAV side, and $\mathbf{f}_B$ is the transmit beamforming vector at BS side.  The transmit/receive beamforming vectors  are  of the form
\begin{subequations} \label{BFvector}
\begin{align}
\mathbf{f}_{B}(\breve{\Psi}_B, \breve{\Omega}_B) =  \frac{1}{\sqrt{N_B}}\mathbf{v}(\breve{\Psi}_B, N_{B,x}) \otimes  \mathbf{v}(\breve{\Omega}_B, N_{B,z}) \\
\mathbf{m}_{U}(\breve{\Psi}_U, \breve{\Omega}_U) =  \frac{1}{\sqrt{N_U}}\mathbf{v}(\breve{\Psi}_U, N_{U,x}) \otimes  \mathbf{v}(\breve{\Omega}_U, N_{U,y})
\end{align}
\end{subequations}
where $(\breve{\Psi}_U, \breve{\Omega}_U)$ and $(\breve{\Psi}_B, \breve{\Omega}_B)$ denote the estimated two-dimensional AoA/AoD.

The free-space path loss of UAV mmWave communications with beamforming is written as  \cite{goldsmith2005wireless}
\begin{align}
P_L(dB) = 10\log_{10} \frac{P_t}{P_r} = -10\log_{10} \frac{G_l \lambda^2}{(4\pi d_{BU})^2}
\end{align}
where $P_t$ is transmit power, $P_r$ is receive power, and  $G_l = \left| \mathbf{m}_{U}^H \mathbf{v}_U \mathbf{v}_B^H \mathbf{f}_{B} \right|^2$ is the product of the transmit and receive antenna field radiation patterns in the LoS direction, i.e., beamforming gain.

To explore the role that  $(\hat{\Psi}_B, \hat{\Omega}_B)$ and  $(\hat{\Psi}_U, \hat{\Omega}_U)$ can play in UAV beamforming, we investigate the free space pathloss for the following three beamforming schemes.

\noindent \emph{\underline{Beamforming Scheme 1:}} At BS side, $\breve{\Psi}_B = \Psi_B, \breve{\Omega}_B = \Omega_B$, and at UAV side, $\breve{\Psi}_U = \Psi_U, \breve{\Omega}_U = \Omega_U$.

\noindent \emph{\underline{Beamforming Scheme 2:}} At BS side, $\breve{\Psi}_B = \hat{\Psi}_B, \breve{\Omega}_B = \hat{\Omega}_B$, and at UAV side, $\breve{\Psi}_U = \Psi_U, \breve{\Omega}_U = \Omega_U$.

\noindent \emph{\underline{Beamforming Scheme 3:}} At BS side, $\breve{\Psi}_B = \hat{\Psi}_B, \breve{\Omega}_B=\hat{\Omega}_B$, and at UAV side, $\breve{\Psi}_U = \hat{\Psi}_U, \breve{\Omega}_U = \hat{\Omega}_U$.

\figref{Pathloss} depicts the pathloss of the three beamforming schemes in the three scenarios proposed in Section II.C.  Beamforming scheme 1 is the benchmark, where the accurate $(\Psi_B, \Omega_B),  ({\Psi}_U, \Omega_U)$ are known, and it always achieves the lowest pathloss. In Beamforming scheme 2, accurate $({\Psi}_U, \Omega_U)$ is known at UAV side, and $(\hat{\Psi}_B, \hat{\Omega}_B)$ at BS side is obtained from UAV navigation system. It can be seen that Beamforming scheme 2  achieves almost the same pathloss as Beamforming scheme 1 (the difference is generally less than $0.5$dB), which indicates that UAV location estimation error incurs negligible performance degradation. By contrast, in Beamforming scheme 3, where both $(\hat{\Psi}_U, \hat{\Omega}_U)$ at UAV side and $(\hat{\Psi}_B, \hat{\Omega}_B)$ at BS side are obtained from UAV navigation system, the pathloss fluctuates dramatically. From  \figref{Pathloss}, we can see that the highest instantaneous path loss is more than $100$dB, while the lowest instantaneous path loss is $56.8$dB. It indicates that the instantaneous performance degradation caused by UAV attitude estimation error can be as significant as $50$dB, which will undoubtedly result in outage.

According to the above results, the roles of UAV navigation information are summarized as follows.
\begin{itemize}
\item Navigation-system-yielded $(\hat{\Psi}_B, \hat{\Omega}_B)$ at BS side can be readily applied to UAV beamforming, which will cause negligible performance loss.
\item Navigation-system-yielded $(\hat{\Psi}_U, \hat{\Omega}_U)$ at UAV side should be further refined by beam training. However, it can used as the prior knowledge for beam training to narrow down the sensing range.
\end{itemize}

\subsection{Framework of the Proposed UAV Beam Training Scheme}

With the navigation-system-yielded $(\hat{\Psi}_B, \hat{\Omega}_B)$ and  $(\hat{\Psi}_U, \hat{\Omega}_U)$, we customize a beam training scheme for UAV mmWave communications and the procedures of which are presented as follows.
\\
\underline{\emph{Step 1.}} Navigation system generates estimated UAV position  $(\hat{x}, \hat{y}, \hat{z})$ and UAV attitude $(\hat{\alpha}, \hat{\beta}, \hat{\gamma})$;\\
\underline{\emph{Step 2.}} With $(\hat{x}, \hat{y}, \hat{z})$ and $(\hat{\alpha}, \hat{\beta}, \hat{\gamma})$, obtain the two-dimensional AoA/AoD $(\hat{\Psi}_B, \hat{\Omega}_B)$ at BS side and $(\hat{\Psi}_U, \hat{\Omega}_U)$ at UAV side according to \eqref{BSside} and \eqref{AngleUAV};\\
\underline{\emph{Step 3.}} Generate the direction-constrained random sensing matrix ${\mathbf{M}}_U$ that is centered at $(\hat{\Psi}_U, \hat{\Omega}_U)$.\\
\underline{\emph{Step 4.}} Channel measurement, i.e., BS performs directional transmit beamforming according to $(\hat{\Psi}_B, \hat{\Omega}_B)$, and UAV sequentially applies the sensing vector $\mathbf{m}_{U,n}, n=1,\cdots, N$ to analog phase shifters to collect channel measurements $\mathbf{y}$. \\
\underline{\emph{Step 5.}} With $\mathbf{y}$ and ${\mathbf{M}}_U$, perform MLE based two-dimensional AoA/AoD estimation to yield $(\tilde{\Psi}_U, \tilde{\Omega}_U)$.

 For illustrative purposes, the design flow is also shown in \figref{Flow}.  Since Step 1 \& 2 have already been explained, in the following subsections, we will primarily introduce Step 3-5.

\begin{figure}[tp]{
\begin{center}{\includegraphics[width=8cm ]{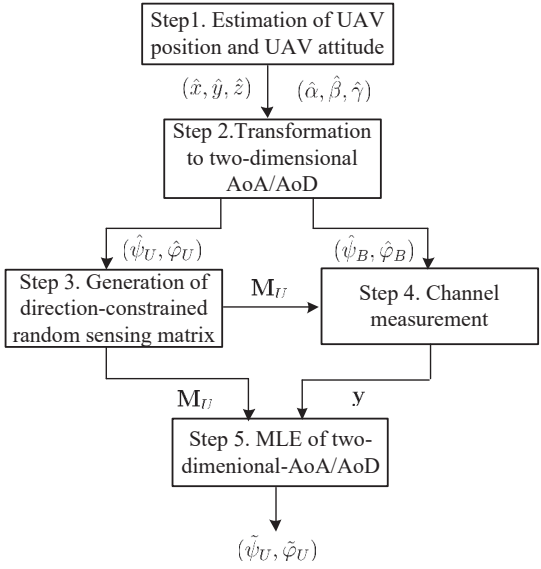}}
\caption{Flow chart of UAV beam training}\label{Flow}
\end{center}}
\end{figure}

\subsection{CS Based Beam Training with Partially Random (Direction-Constrained) Sensing Matrix}
Recall the received signal model in \eqref{sensing}. Since it has already been shown that $(\hat{\Psi}_B, \hat{\Omega}_B)$ derived from the feedback of UAV position causes negligible performance degradation, it is unnecessary to estimate $( {\Psi}_B,  {\Omega}_B)$ at beam training stage. Thus, BS can directly adopt the following transmit beamforming vector throughout the training process
\begin{align}
&\mathbf{f}_{B,n} \equiv \mathbf{f}_{B} = \frac{1}{\sqrt{N_B}}\mathbf{v}(\hat{\Psi}_B, N_{B,x}) \otimes  \mathbf{v}(\hat{\Omega}_B, N_{B,z}),  \notag \\
&\qquad \qquad \qquad \qquad \forall n \in \{ 1,\cdots, N\}
\end{align}
where $\hat{\Psi}_B, \hat{\Omega}_B $ are derived  by substituting $(\hat{x}, \hat{y}, \hat{z})^T$  into \eqref{DirectionEq} and \eqref{BSside}. While, the randomly generated receive beamforming (sensing) vector $\mathbf{m}_{B,n} \in \mathcal{M}$ changes with time $n$. By concatenating $N$ channel measurements, the vector-form received signal is given by
\begin{align}
\mathbf{y}  & = \mathbf{M}_{U}^H \mathbf{H} \mathbf{f}_{B} + \mathbf{w} \notag \\
& = \bar{\tau} \cdot \mathbf{M}_{U}^H \underbrace{\Big(  \mathbf{v}(\Psi_U, N_{U,x}) \otimes  \mathbf{v}(\Omega_U, N_{U,y}) \Big)}_{\mathbf{b}(\Psi_U, \Omega_U)} + \mathbf{w}
\end{align}
where the direction-constrained random sensing  matrix at UAV side is $\mathbf{M}_{U} = [\mathbf{m}_{U,1}, \cdots, \mathbf{m}_{U,N} ]$ and the generation method of which will be introduced in the next subsection. In addition,  $\mathbf{w}  = [ \mathbf{m}_{U,1}^H \bar{\mathbf{w}}_1, \cdots,  \mathbf{m}_{U,N}^H \bar{\mathbf{w}}_N]^T \sim \mathcal{CN}( \mathbf{0}_{N\times 1}, \sigma^2 \mathbf{I}_{N\times N}) $ is the equivalent complex additive Gaussian white noise, and the coefficient
\begin{align}
 \bar{\tau} =& \frac{\lambda e^{j\left( \frac{-2\pi  d_{BU}  } {\lambda } + \frac{ \pi(\Psi_B - \hat{\Psi}_B)(N_{B,x}-1)}{2} + \frac{ \pi(\Omega_B - \hat{\Omega}_B)(N_{B,z}-1)}{2}  \right)  }}{4\pi d_{BU}\sqrt{N_{B,x}N_{B,z}}}  \cdot \notag \\
&    \left(\frac{ \sin ( \frac{\pi   (\Psi_B - \hat{\Psi}_B) N_{B,x}}{2})  }{\sin ( \frac{\pi (\Psi_B - \hat{\Psi}_B)}{2}  )  }\right) \left(\frac{ \sin ( \frac{\pi   (\Omega_B - \hat{\Omega}_B)  N_{B,z}}{2})  }{\sin ( \frac{\pi  (\Omega_B - \hat{\Omega}_B)}{2}  )  }\right) \label{EqChannel}
\end{align}
where the sinc functions (second and third term of \eqref{EqChannel}) are related to the power loss caused by  the estimation error of $(\hat{\Psi}_B, \hat{\Omega}_B)$, which is already shown to be negligible in \figref{Pathloss}.

With the measurement signal vector $\mathbf{y}$, the two-dimensional AoA (as it is in downlink transmission) at UAV side can be obtained through maximum likelihood estimation (MLE), and the MLE problem is equivalently written as \cite{wang2020joint}
\begin{align}
 &P1: \max_{\Psi_U, \Omega_U}  \left\| \frac{\mathbf{b}(\Psi_U, \Omega_U)^H \mathbf{M}_U}{ \left\|  \mathbf{M}_U^H\mathbf{b}(\Psi_U, \Omega_U) \right\|_2}\mathbf{y} \right\|_2^2 \notag \\
&\;\;s.t. \; -1\leq \Psi_U < 1 \notag\\
& \;\;\;\;\;\;\;\; -1\leq \Omega_U <1 \notag
\end{align}

Similar to  \cite{wang2020joint}, P1 can be solved via the following two steps.

\noindent \textbf{Step 1. Coarse Search}

Discretize   $\Psi_U$ and $\Omega_U$  with two-dimensional finite granularity (with quantization levels $Z_{\Psi_U}$ and $Z_{\Omega_U}$) and then exhaustively search for the $N_{pk}$ largest maxima.

\noindent \textbf{Step 2. Fine Search}

For a  discrete maximum $(\check{\Psi}_U, \check{\Omega}_U)^T$, run gradient descent search starting from $( \Psi_U^{(1)},  \Omega_U^{(1)})^T =(\check{\Psi}_U, \check{\Omega}_U)^T $ as follows
\begin{align} \label{FineSearch}
 \left(
   \begin{array}{c}
     \Psi_U^{(i+1)} \\
     \Omega_U^{(i+1)} \\
   \end{array}
 \right)  =    \left(
   \begin{array}{c}
     \Psi_U^{(i)} \\
     \Omega_U^{(i)} \\
   \end{array}
 \right)    \oplus \lambda    \left(
   \begin{array}{c}
     \frac{\partial g({\Psi_U, \Omega_U})}{\partial \Psi_U} \big|_{\Psi_U = \Psi_U^{(i)}} \\
     \frac{\partial g({\Psi_U, \Omega_U})}{\partial \Omega_U} \big|_{\Omega_U = \Omega_U^{(i)}} \\
   \end{array}
 \right)
\end{align}
where  the expressions of partial derivatives are given in Appendix A, and $\lambda$ is the preset step size. The iterative process stops when $(\Psi_U^{(i+1)} \ominus \Psi_U^{(i)} )^2 + (\Omega_U^{(i+1)} \ominus \Omega_U^{(i)} )^2 \leq \epsilon$, and $\epsilon$ is a preset parameter.

Repeat the above operations over the rest $N_{pk}-1$ maxima derived in Step 1, and select the best one as the estimated two dimensional AoA $(\tilde{\Psi}_U, \tilde{\Omega}_U)$.

\begin{remark}
The  complexity of Step 1 is $\mathcal{O}(2^2 Z_{\Psi_U} Z_{\Omega_U})$. The  complexity of Step 2 mainly arises from the computation of the gradients  $ \frac{\partial{g(\Psi_U, \Omega_U)}}{\partial{\Psi_U}}$ and $\frac{\partial{g(\Psi_U, \Omega_U)}}{\partial{\Omega_U}} $, which, according to Eq. \eqref{PsiUpartial} and Eq. \eqref{OmegaUpartial},  is $\mathcal{O}(N_{U,x}N_{U, y} N)$. Hence, the complexity of Step 2 is $\mathcal{O}(n_{iter}N_{pk}N_{U, x}N_{U, y}N)$, where the iteration number $n_{iter}$ depends on step size and stopping criterion of the gradient method and is generally less than $20$; $N_{U,x}$, $N_{U, y}$ are the element antenna numbers of UPA along x axis and y axis, and $N$ is the training length. Thus, the overall complexity is $\mathcal{O}(2^2 Z_{\Psi_U} Z_{\Omega_U} + n_{iter}N_{pk}N_{U,x}N_{U, y}N)$.
\end{remark}

\subsection{Generation of Partially Random (Direction-Constrained) Sensing Matrix}

The key part of our scheme is the direction-constrained sensing matrix $\mathbf{M}_{U}$.  As mentioned in Section IV. A, conventional generation methods of random sensing matrix will result in  omnidirectional sensing range. To exploit navigation information,  we need to design the direction-constrained  random sensing matrix $\mathbf{M}_{U} $ given $\hat{\Psi}_U, \hat{\Omega}_U $.

Without loss of generality, the sensing matrix $\mathbf{M}_U$ is constructed column by column. A traditional way to construct the sensing vector $\mathbf{m}_U$ is to set
\begin{align}
\mathbf{m}_U(n_U) = \frac{1}{\sqrt{N_U}}e^{j \pi \varphi_{n_U}},\;\; n_U = 1,2, \cdots, N_U  \notag
\end{align}
and then to generate $N_U$ i.i.d. (independent and identically distributed) random phase coefficients
\begin{align}
\varphi_{n_{U}} \sim U(-1, 1), \;\; n_U = 1,2, \cdots, N_U \notag
\end{align}
As all the $ N_U$ phase coefficients are random, we name the generated $\mathbf{M}_U$ that follows this method as \emph{fully random} sensing matrix. Due to its full randomness, $\mathbf{M}_U$ is omnidirectional with a high probability\cite{myers2020deep}.

However, beam training with an omnidirectional sensing matrix is inefficient for UAV mmWave communications,  since the navigation system of UAV is capable of providing a rough estimate of the two-dimensional AoA/AoD, which can be utilized to narrow down the search interval. To this end, we propose to design the direction-constrained (partially) random sensing matrix $\mathbf{M}_U$.

\begin{figure}[tp]{
\begin{center}{\includegraphics[width=8.5cm ]{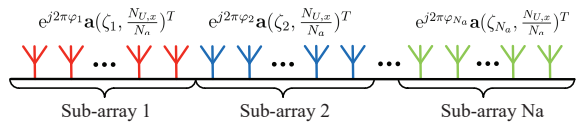}}
\caption{Visualization of sub-array partition along one dimension}\label{Subarray}
\end{center}}
\end{figure}

 \begin{figure*}[tp]{
\begin{center}{\includegraphics[width=14cm ]{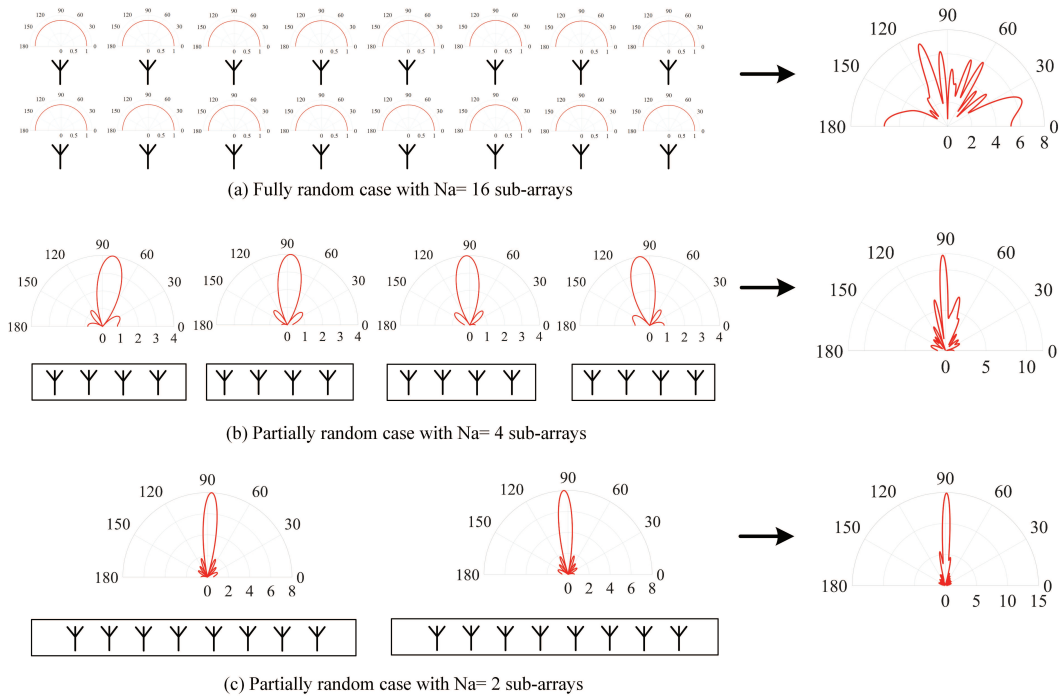}}
\caption{Visualization of sub-array based design of random sensing vector (the scale of  angle is $\cos^{-1} \Psi_U$)}\label{PatternVisual}
\end{center}}
\end{figure*}
 We decompose the sensing vector $\mathbf{m}_U$ into two dimensions ($x$ axis and $y$ axis), i.e.,
\begin{align}
\mathbf{m}_U = \mathbf{m}_{U, x} \otimes \mathbf{m}_{U, y} \label{KronDecomp}
\end{align}
Note that the UPA can be treated as $N_{U,y}$ uniform linear arrays (ULAs) placed along  $x$ axis, and their receive beamforming patterns are determined by the vector $ \mathbf{m}_{U, x}$; Similarly, the UPA can also be treated as $N_{U,x}$ ULAs placed along $y$ axis, whose receive beamforming patterns are determined by the vector $ \mathbf{m}_{U, y}$. Therefore,  the generation of $\mathbf{m}_{U}$ can be decoupled into two parallel subproblems that generate $\mathbf{m}_{U, x}$ and $\mathbf{m}_{U, y}$, separately.

Take $\mathbf{m}_{U, x}$ as an example, to manage its radiation pattern, we will sacrifice part of its randomness.  Specifically, we partition the ULA (along $x$ axis) into $N_a$ sub-arrays as shown in \figref{Subarray}, and let $\mathbf{m}_{U, x}$ follow the format of
\begin{align}
\mathbf{m}_{U,x} &= \frac{1}{\sqrt{N_{U,x}}}\left[e^{j \pi \varphi_{1}} \mathbf{v}^T(\zeta_{1}, \frac{N_{U,x}}{N_a}),\; e^{j \pi \varphi_{2}}\mathbf{v}^T(\zeta_{2}, \frac{N_{U,x}}{N_a}),\; \cdots,\;  \right. \notag\\
&\qquad \qquad \qquad \qquad \left. e^{j  \pi \beta_{N_a}}\mathbf{v}^T(\zeta_{N_a}, \frac{N_{U,x}}{N_a}) \right]^T \label{SubEqA}
\end{align}
where the term $ \mathbf{v}(\zeta_{n_a}, \frac{N_{U,x}}{N_a}), n_a = 1,\cdots, N_a$ determines the beam pattern of each sub-array.

It is noteworthy that the controllable variables are the phase coefficient $\varphi_{n_a}$ and the center angle $\zeta_{n_a}$ of each sub-array\footnotemark.

1) The phase coefficients $\varphi_{n_a}, \; (n_a = 1,2, \cdots, N_a)$ determine the collective effects of each sub-array. Similar to fully random case, we generate the i.i.d.  $\varphi_{n_a}, (n_a = 1, 2, \cdots, N_a)$ with distribution  $U(-1,1)$.

2) As for the center angle $\zeta_{n_a}$, it is related to the angle range of the sensing matrix $\mathbf{M}_{U}$.  The beam width of each sub-array with the beamforming vector $\mathbf{a}(\zeta_{n_a}, \frac{N_{U,x}}{N_a})$ is $\frac{2N_a}{N_{U,x}}$ \cite{tse2005fundamentals}, which indicates that the majority of its received signal power is within the range of $\left(-\frac{N_a}{N_{U,x}} + \zeta_{n_a}, \frac{N_a}{N_{U,x}}  + \zeta_{n_a} \right)$. To constrain the range of direction and in the meantime guarantee randomness, we generate the i.i.d. $\zeta_{n_a}, (n_a = 1, 2, \cdots, N_a)$ with the distribution  $U(\zeta_{Low}, \zeta_{Upp})$. Considering the transition bands, i.e., $\left(\zeta_{Low} - \frac{N_a}{N_{U,x}}, \zeta_{Low}\right)$ and $\left(\zeta_{Upp}, \zeta_{Upp} + \frac{N_a}{N_{U,x}}\right)$, the direction range of the generated sensing matrix is thus $\left(\zeta_{Low} - \frac{N_a}{N_{U,x}}, \zeta_{Upp} + \frac{N_a}{N_{U,x}}\right)$

As part of the degree of freedom is sacrificed, the generated random sensing matrix is termed as  \emph{partially random}.

\begin{remark}
In practice, a set of  sensing matrices with different sensing ranges can be stored in advance, and the best one will be adaptively selected according to the degree of AoA/AoD fluctuation at UAV side.
\end{remark}

\subsection{Generalization of Fully Random (Omnidirectional) and Partially Random (Direction-Constrained) Sensing Matrices}

\begin{figure*}[t]
\begin{minipage}[!h]{0.32\linewidth}
\centering
\includegraphics[ width=1.05\textwidth]{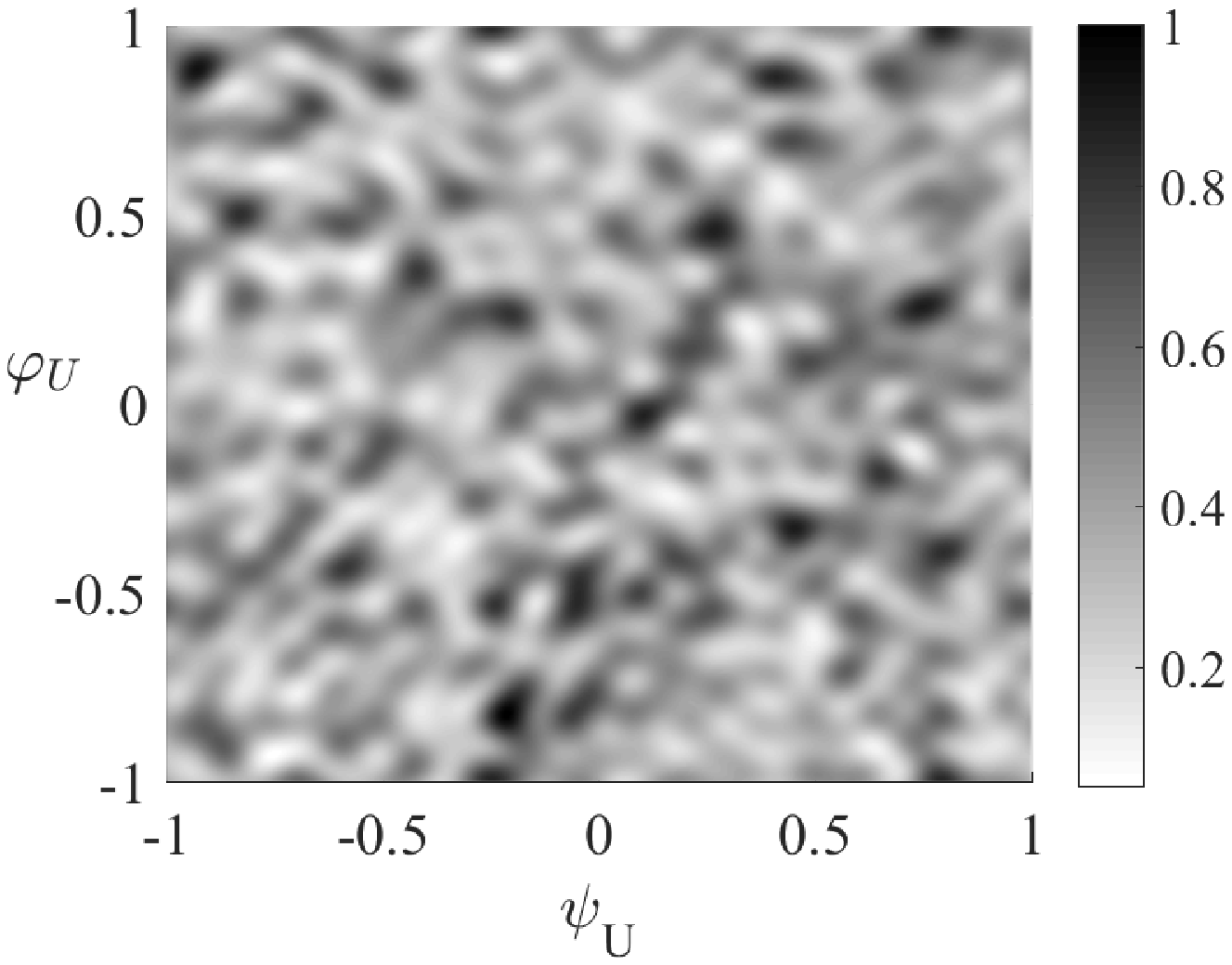}
\subcaption{\scriptsize Fully random case with $N_a=16$ sub-arrays, and the sensing range is $ \Psi_U, \Omega_U \in (-1, 1)$}
\label{Scenario1}
\end{minipage}
\begin{minipage}[!h]{0.32\linewidth}
\centering
\includegraphics[ width=1.05\textwidth]{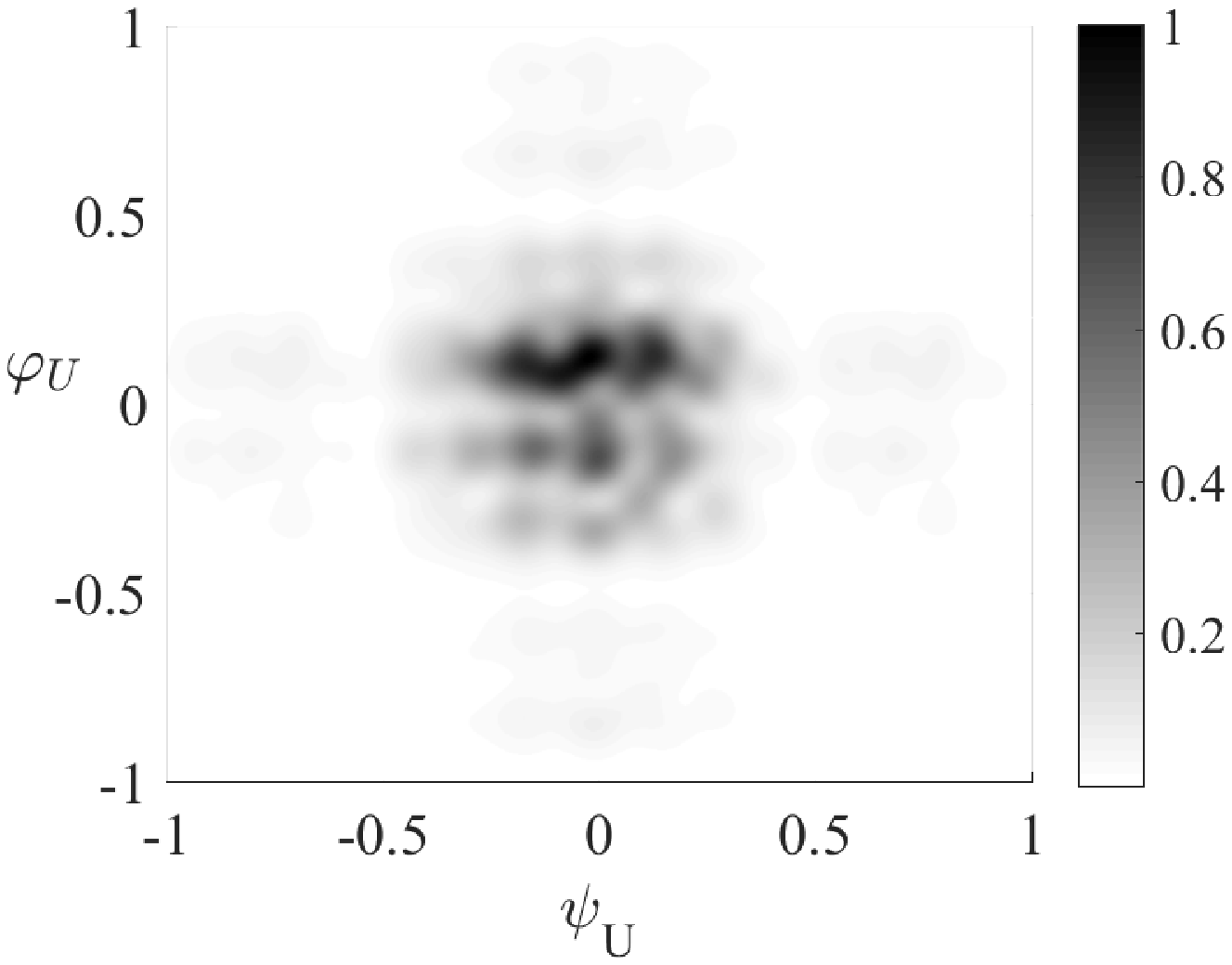}
\subcaption{\scriptsize Partially random case with $N_a = 4$ sub-arrays, and the sensing range is $\Psi_U, \Omega_U \in (-0.4, 0.4)$}
\label{Scenario2}
\end{minipage}
\begin{minipage}[!h]{0.32\linewidth}
\centering
\includegraphics[ width=1.05\textwidth]{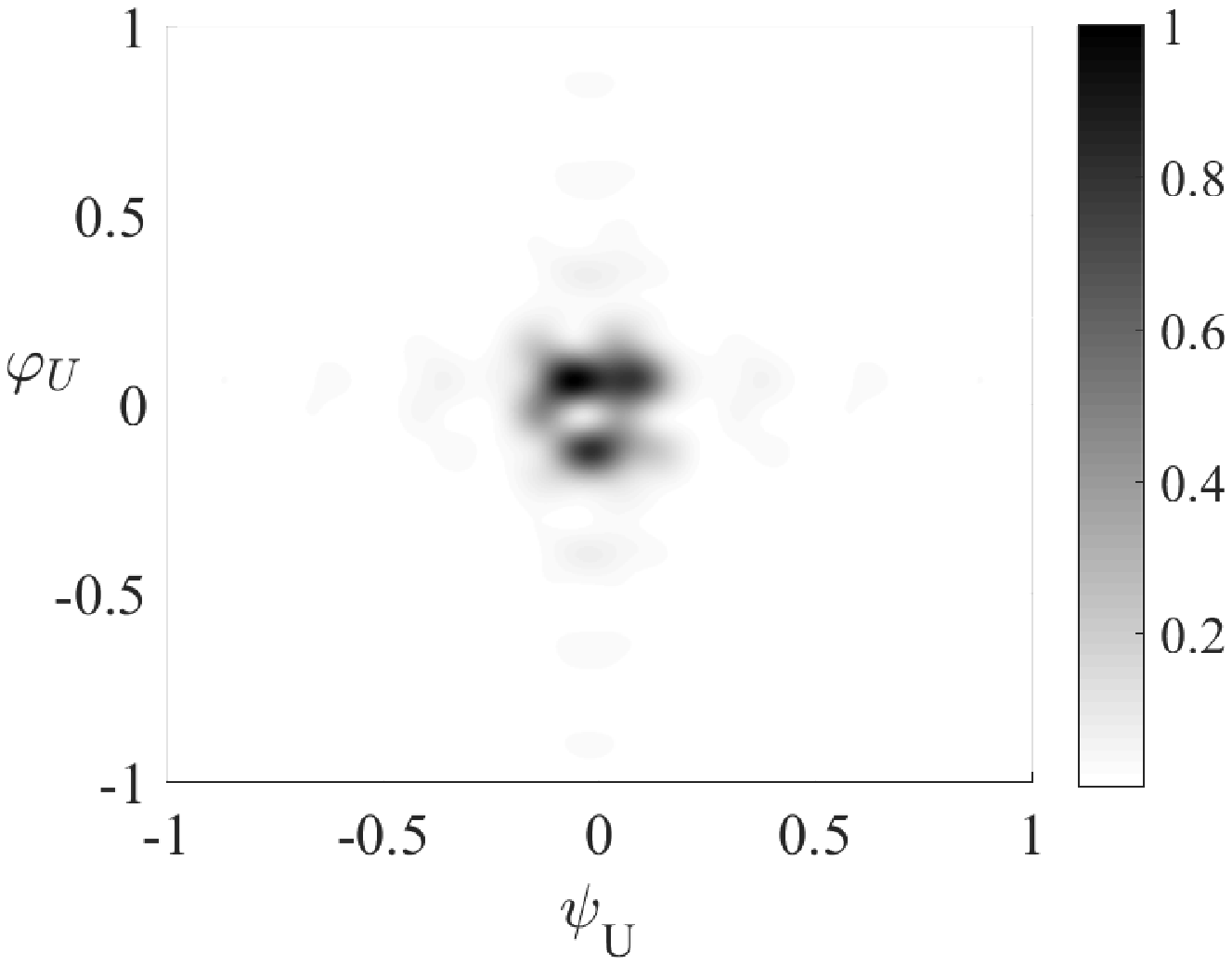}
\subcaption{\scriptsize Partially random case with $N_a=2$ sub-arrays, and the sensing range is $\Psi_U, \Omega_U \in (-0.225, 0.225)$}
\label{Scenario3}
\end{minipage}
\caption{Beam space of fully and partially random sensing matrices} \label{RandomMatrix}
\end{figure*}

To clarify the connection between the proposed partially random sensing matrix and the conventional fully random sensing matrix, we make a generalization in this subsection.

For a fully random sensing matrix, the receive beamforming vector that corresponds to any component ULA of the UPA, either along $x$ axis and $y$ axis, is fully random. Without loss of generality,  we denote the receive beamforming vector of a typical component ULA as
\begin{align}
\mathbf{m}_{U,x,1} = \frac{1}{\sqrt{N_{U,x}}} \left[e^{j \pi \bar{\varphi}_{1}}, e^{j \pi \bar{\varphi}_{2}}, \cdots, e^{j \pi \bar{\varphi}_{N_{U,x}}} \right]^T   \label{BFux}
\end{align}
where the phase coefficients $\bar{\varphi}_{1}, \cdots, \bar{\varphi}_{N_{U,x}} \sim U(-1,1)$ are i.i.d.  In accordance with \eqref{SubEqA}, $\mathbf{m}_{U,x,1}$ can also be written as
\begin{align}
\mathbf{m}_{U,x,1} &= \frac{1}{\sqrt{N_{U,x}}}\left[e^{j \pi \bar\varphi_{1}} \mathbf{v}(\bar\zeta_{1}, 1)^T,\; e^{j \pi \bar\varphi_{2}}\mathbf{v}(\bar\zeta_{2}, 1)^T,  \; \cdots,\; \right. \notag\\
&\qquad \qquad \qquad \qquad \left. e^{j  \pi \bar\varphi_{N_{U,x}}}\mathbf{v}(\bar\zeta_{N_{U,x}}, 1)^T \right]^T  \label{BFux2}
\end{align}
where, based on \eqref{UavULA1}, $\mathbf{v}(\bar\zeta_{1}, 1) = \cdots = \mathbf{v}(\bar\zeta_{N_{U,x}}, 1) \equiv 1$ regardless of the center angle $\bar\zeta_{n_{U,x}}$.

Therefore, from a generalized perspective, the fully random method partitions the ULA into $N_a = N_{U,x}$ sub-arrays and each sub-array consists of only one element antenna,  whose receive beamforming pattern is omnidirectional \cite{tse2005fundamentals}. To obtain further insights into the proposed random matrix generation method, we plot the beam patterns of each component sub-array and the constituted ULA in \figref{PatternVisual} when $N_{U,x}=16$ in the following three cases. \vspace{0.01cm}\\
(1) \underline{\emph{Fully random case with $N_a = 16$}}, in which the radiation pattern of each sub-array is omnidirectional. Thus, the beamforming vector of the constituted ULA is merely decided by $\varphi_{n_a}$ and its beamforming pattern is statistically omnidirectional. \\
(2) \underline{\emph{Partially random case with $N_a = 4$}}, in which each sub-array consists of $4$ element antennas. With the beamforming vector $\mathbf{v}(\zeta_{n_a}, 4)$, each sub-array is indeed a directional antenna whose beamwidth is $\frac{2N_a}{N_{U,x}} = 0.5$ and center angle is $\zeta_{n_a}$. The beamforming vector of the constituted ULA is  jointly decided by $\zeta_{n_a}$ and $\varphi_{n_a}$. We set $\zeta_{n_a} \sim U(-0.15, 0.15)$, thus the beamforming pattern of the constituted ULA is random within the angle range $(-0.4, 0.4)$. \\
(3) \underline{\emph{Partially random case with $N_a = 2$}}, in which each sub-array consists of $8$ element antennas. With the beamforming vector $\mathbf{v}(\zeta_{n_a}, 8)$, each sub-array is indeed  a directional antenna whose beamwidth is $\frac{2N_a}{N_{U,x}} = 0.25$ and center angle is $\zeta_{n_a}$. We set $\zeta_{n_a} \sim U(-0.1, 0.1)$, thus the beamforming pattern of the constituted ULA is random within the angle range $(-0.225, 0.225)$. \\

To visualize the sensing range of the generated $\mathbf{M}_U$, we plot its beam space in \figref{RandomMatrix}, where the grayscale is proportional to
\begin{align}
G = \frac{\left\| \mathbf{M}_U^H \mathbf{b}(\Psi_U, \Omega_U) \right\|_2^2}{\max_{\Psi_U, \Omega_U} \left\| \mathbf{M}_U^H \mathbf{b}(\Psi_U, \Omega_U) \right\|_2^2}
\end{align}
As can be seen that for (1) Fully random case with $N_a = 16$, the sensing range is $-1<\Psi_U<1, -1<\Omega_U<1$, which verifies its omnidirectional property; for (2) Partially random case with $N_a = 4$, the sensing range is $-0.4<\Psi_U<0.4, -0.4<\Omega_U<0.4$, which is in accordance with our expectation; for (3) Partially random case with $N_a = 2$, the sensing range is $-0.225<\Psi_U<0.225, -0.225<\Omega_U<0.225$ as expected.

\begin{remark}
With the rough AoA estimated by UAV navigation system, beam training with the partially random matrix strikes a balance between beamforming gain, which is essential for mmWave communications, and randomness, which is required by compressed sensing.
\end{remark}

\section{Numerical Results}

In this section, we  numerically study the performance of our proposed UAV beam training scheme.

We set the simulation parameters as follows. The speed of light is $3\times 10^8$m/s, the operating frequency is $f_c= 28$GHz, the noise power at UAV side is $-84$dBm, and UAV hovers in a hemisphere that is $200$ meters away from BS. The size of the UPA at BS side is  $N_{B,x} = N_{B,z} =  16$, the size of the UPA at UAV side is $N_{U,x} = N_{U,y}=16$. The errors of the estimated yaw angle, pitch angle, and roll angle from navigation system are i.i.d. Gaussian variables  with the mean values being $0$, and the standard deviations being $\sigma_\alpha = \sigma_\beta = \sigma_\gamma = 0.05$. The errors of UAV position coordinate is  $\mathbf{n}\sim \mathcal{N}(\mathbf{0}_{3\times 1}, \mathbf{I}_{3\times 3})$.

\begin{figure}[tp]{
\begin{center}{\includegraphics[width=8cm ]{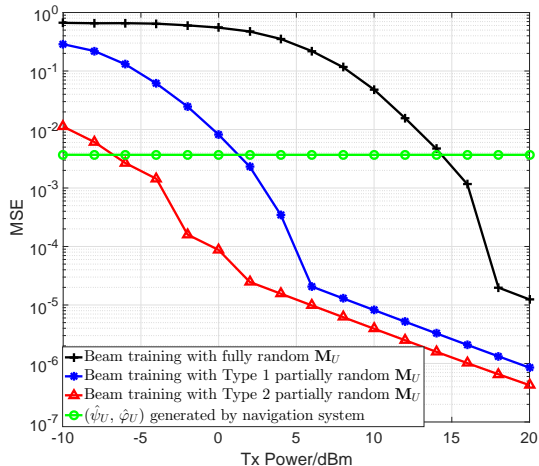}}
\caption{Estimation accuracy of the MSE of the two dimensional AoA at UAV side when training length $N=6$}\label{MSEcomp}
\end{center}}
\end{figure}

In \figref{MSEcomp}, we study the mean squared error (MSE) performance of the estimated two-dimensional AoA at UAV side, which fundamentally determines the power gain in the subsequent beamforming operation. We set the training length as $N=6$. The four methods in this simulation include beam training with fully random (omnidirectional) sensing matrix and two partially random (direction-constrained) sensing matrices, whose sensing matrices are generated in the same way as Section III.C.  Beam training with fully random (omnidirectional) sensing matrix is a conventional approach that merely depends on channel sounding/measurement and does not utilize navigation information. It is used as a benchmark method in the numerical study. In addition,  the two-dimensional AoA estimated by UAV navigation system, i.e., $(\hat\Psi_U, \hat\Omega_U)$ is used as another benchmark. The performance metric is defined as
\begin{align}
MSE = \mathbb{E}\left((\tilde{\Psi}_U \ominus \Psi_U)^2 + (\tilde\Omega_U \ominus \Omega_U)^2\right),   \notag
\end{align}
which is the MSE of the estimated two-dimensional AoA at UAV side.  The reason that we use ``$\ominus$" operation instead of ``$-$" operation is  that $\left| \mathbf{v}^H(\tilde{\Psi}_U, N_{U,x}) \mathbf{v}({\Psi}_U, N_{U,x}) \right| = \left| \frac{ \sin ( \frac{\pi   (\Psi_U - \tilde{\Psi}_U) N_{U,x}}{2})  }{\sin ( \frac{\pi (\Psi_U - \tilde{\Psi}_U)}{2}  )  } \right|$ is periodic and repeats at intervals of $2$.
In \figref{MSEcomp}, $x$ axis represents transmit power at beam training stage. Thus, MSE of  navigation based method,  which bypasses beam training operation, does not change with $x$.
From the figure, we can see that beam training with partially random (direction-constrained) sensing matrix 2, whose sensing range is $\Psi_U \in (-0.225 + \hat{\Psi}_U, 0.225 + \hat{\Psi}_U), \Omega_U \in (-0.225 + \hat{\Omega}_U, 0.225 + \hat{\Omega}_U)$, achieves the best performance. It becomes better than  navigation based method from $-6$dBm. Beam training with partially random (direction-constrained) sensing matrix 1, whose sensing range is $\Psi_U \in (-0.4 + \hat{\Psi}_U, 0.4 + \hat{\Psi}_U), \Omega_U \in (-0.4 + \hat{\Omega}_U, 0.4 + \hat{\Omega}_U)$, is inferior to  navigation based method in low SNR regimes, but it becomes superior starting from $2$dBm. As for beam training with fully random (omnidirectional) sensing matrix, its MSE outperforms  navigation based method from $16$dBm. The simulation results indicate that by appropriately integrating navigation information and channel sounding/measurement, i.e., narrowing down sensing range, AoA estimation accuracy of the random beamforming based beam training for UAV mmWave communications can be significantly improved.

\begin{figure}[tp]{
\begin{center}{\includegraphics[width=8cm ]{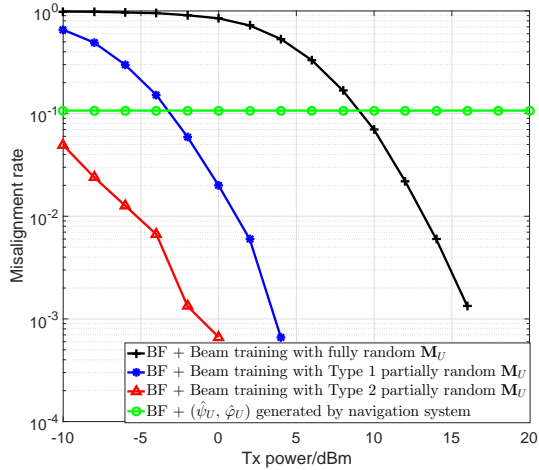}}
\caption{Misalignment rate  of UAV beamforming enabled by the proposed beam training methods when training length $N=6$}\label{Misalignment}
\end{center}}
\end{figure}

Important as it is, the accuracy of AoA/AoD is not the key performance indicator of wireless communications. To obtain further insights into  beam training's impact on the subsequent beamforming operation,  in \figref{Misalignment}, we study the misalignment rate of UAV beamforming based on the AoA/AoD estimated by UAV beam training. Owing to the directional transmission mechanism of mmWave communications, inaccurate AoA/AoD will incur a sharp decline of the received power. Without timely adjustments, e.g., beam direction adjustment, the communication link will suffer from outage. In the simulation,  misalignment rate is defined as
\begin{align}
& Pr(P< P_{TH}) = Pr\Big(20\log_{10}\left|\mathbf{m}_U^H(\tilde{\Psi}_U, \tilde{\Omega}_U)\mathbf{H}\mathbf{f}_B(\hat{\Psi}_B, \hat{\Omega}_B)\right|  \notag \\
&\qquad\qquad\qquad  <  20\log_{10}\left|\frac{\lambda \sqrt{N_{B,x}N_{B,z}N_{U,x}N_{U,y}}}{4\pi d_{BU}}\right| - 10\Big) \notag
\end{align}
where the threshold  $P_{TH}$ is $10$dB less than the maximum received power that is achieved by perfect beamforming (i.e., when $(\hat{\Psi}_B, \hat{\Omega}_B) = ({\Psi}_B, {\Omega}_B) $ and $(\tilde{\Psi}_U, \tilde{\Omega}_U) = ({\Psi}_U, {\Omega}_U) $). It can be seen from the figure that misalignment occurs in $10\%$ of channel realizations when AoA and AoD are both generated by UAV navigation system. This, together with \figref{Pathloss}, indicates that it is almost impossible to maintain a reliable mmWave link merely depending on UAV navigation system. For beam training with fully random (omnidirectional) sensing matrix, which is independent of  navigation information, its performance is inferior to navigation based method when transmit power is from $-10$dBm to $6$dBm. By contrast, beam training scheme with Type 2 partially random (direction-constrained) sensing matrix, which fuses  navigation information and channel measurements, outperforms the two aforementioned benchmark methods. Besides, beam training scheme with Type 1 partially random (direction-constrained) sensing matrix is slightly worse than beam training with Type 2 partially random (direction-constrained) sensing matrix. This is because its sensing range is redundantly wide. However, it is noted that the appropriate sensing range is proportional to the degree of UAV jittering effects. When the resulting fluctuation range of AoA becomes large, Type 1 random sensing matrix could be more suitable than Type 2 random sensing matrix.

\begin{figure*}[t]
\begin{minipage}[!h]{.5\linewidth}
\centering
\includegraphics[width=0.9\textwidth]{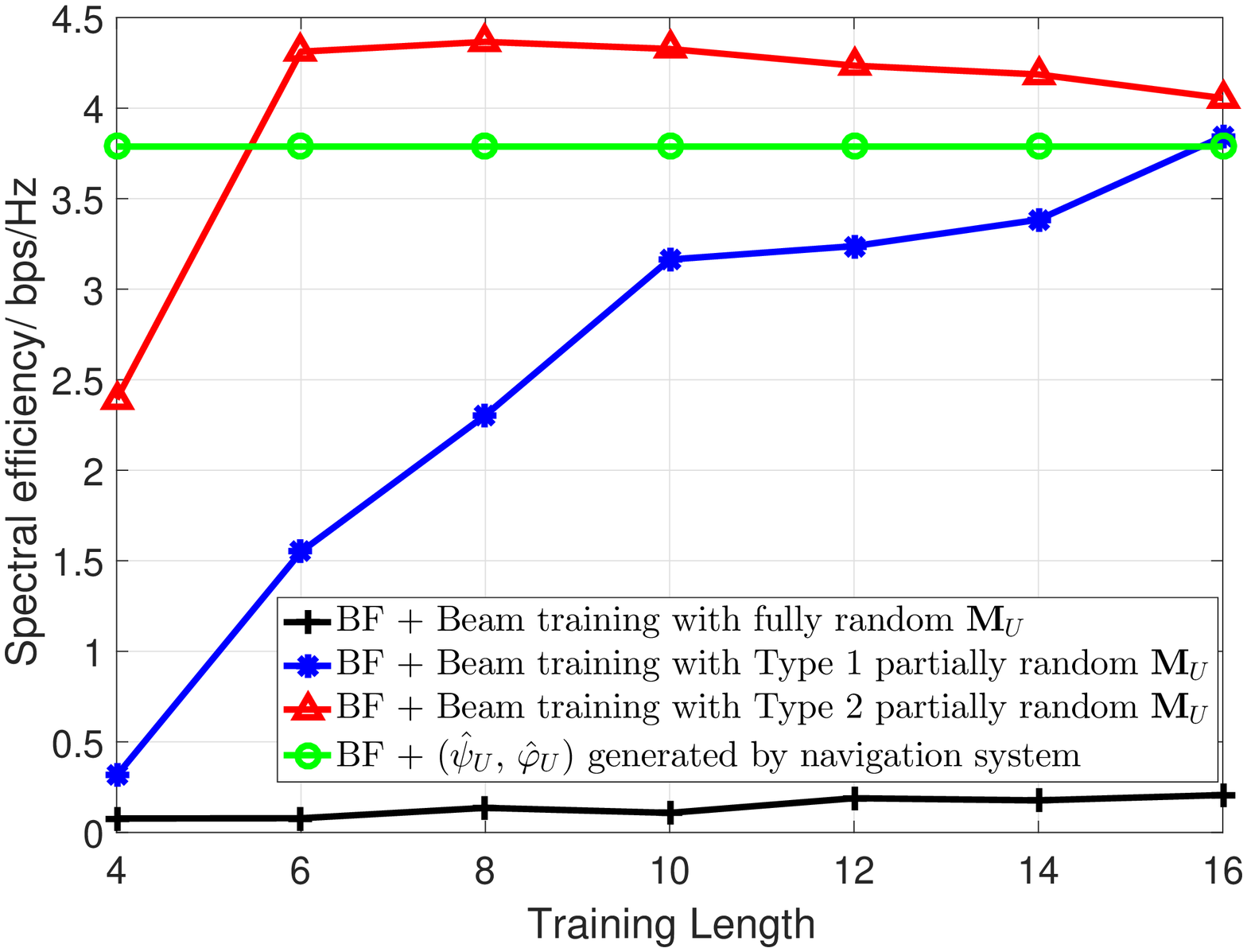}
\subcaption{  When transmit power is $-10$dBm}
\label{SpectralEfficiencySNR1}
\end{minipage}
\begin{minipage}[!h]{.5\linewidth}
\centering
\includegraphics[width=0.9\textwidth]{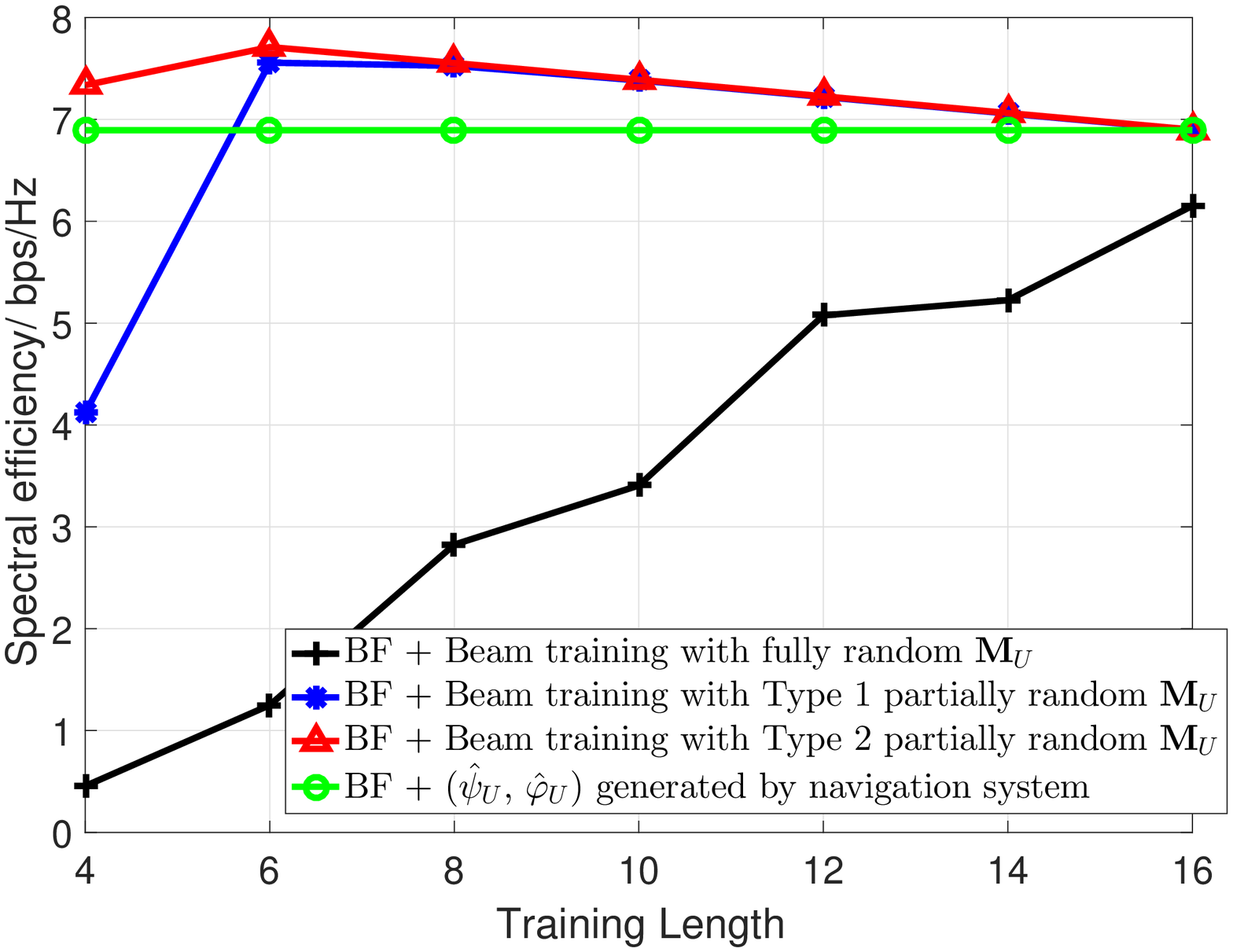}
\subcaption{  When transmit power is $0$dBm}
\label{SpectralEfficiencySNR6}
\end{minipage}
\begin{minipage}[!h]{.5\linewidth}
\centering
\includegraphics[width=0.9\textwidth]{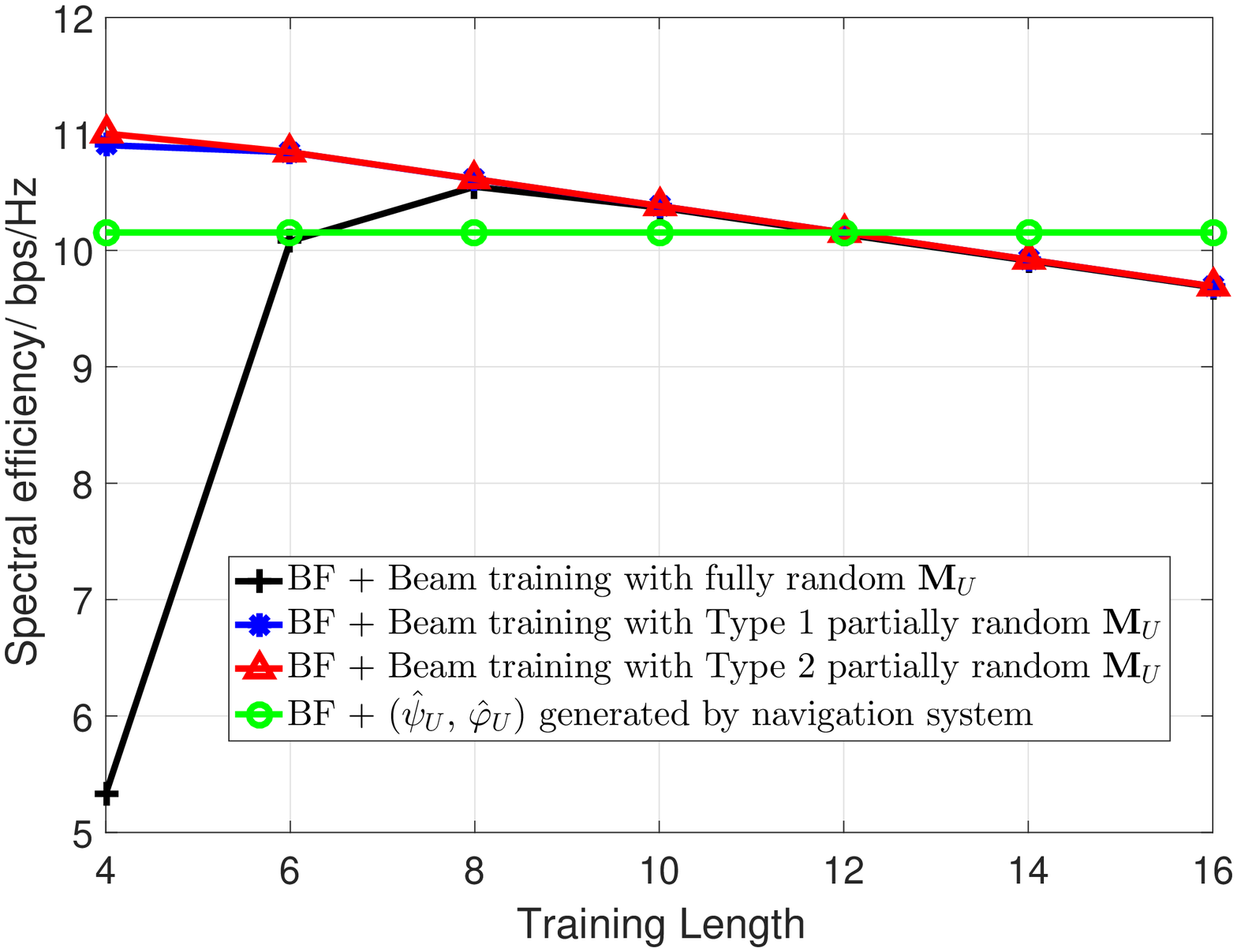}
\subcaption{  When transmit power is $10$dBm}
\label{SpectralEfficiencySNR11}
\end{minipage}
\hspace{0.1cm}
\begin{minipage}[!h]{.5\linewidth}
\centering
\includegraphics[width=0.9\textwidth]{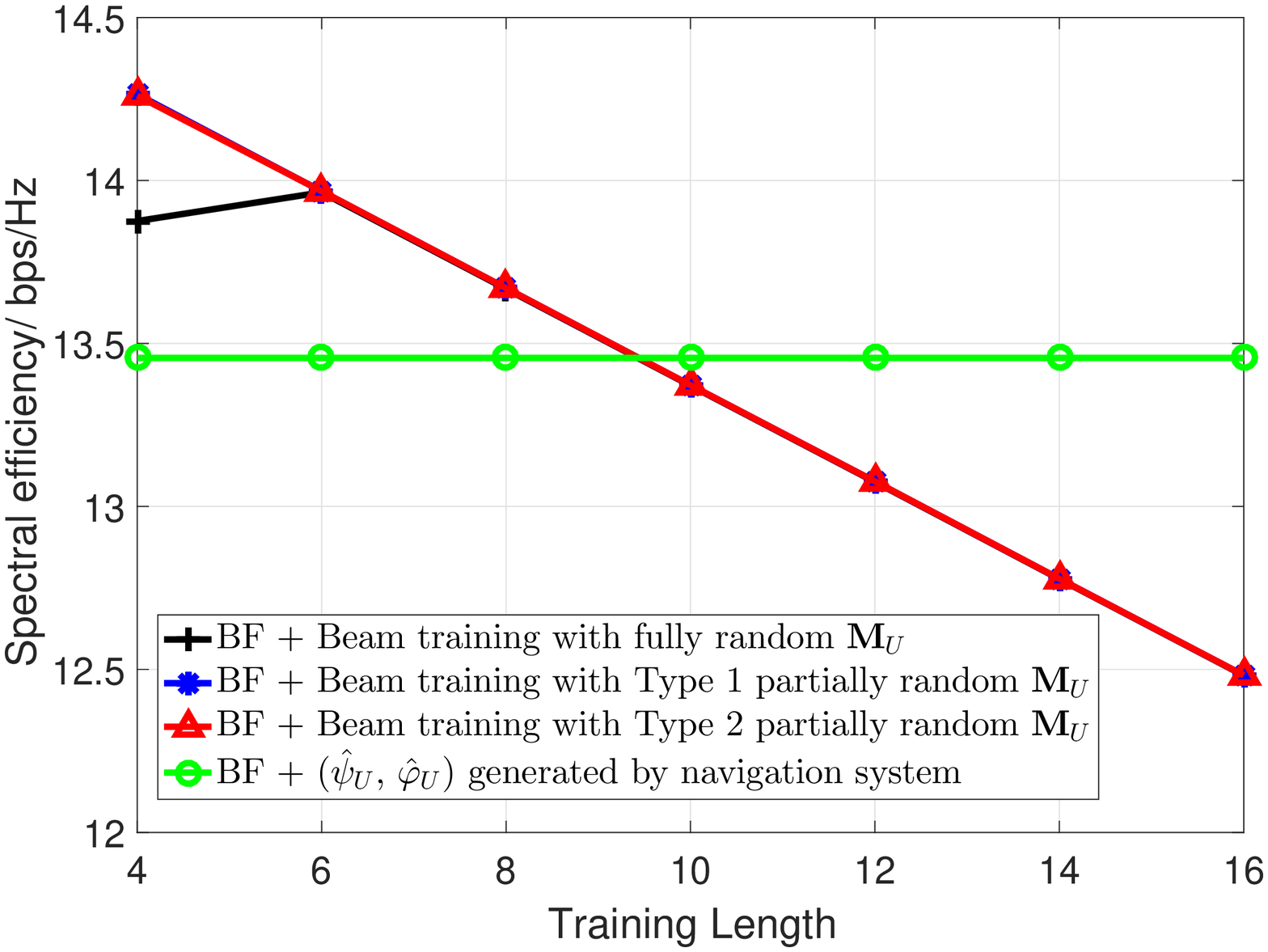}
\subcaption{  When transmit power is $20$dBm}
\label{SpectralEfficiencySNR16}
\end{minipage}
\caption{Spectral efficiency of UAV beamforming enabled by the proposed
beam training methods} \label{TrainingLen}
\end{figure*}

In \figref{TrainingLen}, we study the spectral efficiency of UAV beamforming enabled by the proposed beam training method. We assume that the coherence time of UAV mmWave channel consists of $N_{coh} = 100$ time intervals, and spectral efficiency is defined as
\begin{align}
&SE = \mathbb{E}\left({\frac{R(N_{coh} - N)}{N_{coh}}}\right)
\end{align}
where $R$ is the data rate of UAV beamforming. As beam training operation costs $N$ time intervals, the effective time for data transmission is $N_{coh} - N$.  For beam training,  the larger $N$ will on one hand bring about  better accuracy of AoA estimation  and on the other hand result in  less effective time  for data transmission. By contrast,  navigation based method does not count on channel sounding/measurement, and its effective data transmission time is thus $N_{coh}$.  In \figref{SpectralEfficiencySNR1}, when the transmit power is $-10$dBm, we plot the spectrum efficiency of different methods with the training length  $N = 4, 6, \cdots, 16$. It can be seen that when the training length is $N=8$, the spectral efficiency is maximized for beamforming enabled by beam training with Type 2 partially random sensing matrix, and with the increase of training length $N$, its spectral efficiency gradually decreases. By contrast, beamforming enabled by beam training with Type 1 partially random sensing matrix is not comparable to navigation based method and becomes almost the same only when $N = 16$;  spectral efficiency of beamforming enabled by beam training with fully random sensing matrix is very poor and improves slightly with the increase of training length $N$. It indicates that beam training with an appropriate sensing range is fast and accurate even in low SNR regimes. In \figref{SpectralEfficiencySNR6}, when the transmit power is $0$dBm,  the optimal training length for beam training with  Type 1 partially random  sensing matrix (direction-constrained) achieves almost the same performance as beam training with  Type 2 partially random (direction-constrained) sensing matrix from $N=6$ to $N=16$. Meanwhile, the spectral efficiency of  beam training with fully random (omnidirectional) sensing matrix increases significantly  with the training length $N$. It indicates that beam training with a proper sensing range is more power efficient but would lose its performance advantage when SNR increases. Our conjecture is further verified in \figref{SpectralEfficiencySNR11} and \figref{SpectralEfficiencySNR16}, where we find that beam training with  Type 1 direction-constrained (partially random) sensing matrix and  beam training with omnidirectional (fully random) sensing matrix gradually achieves the same spectral efficiency as  beam training with  Type 2 direction-constrained (partially random) sensing matrix when SNR condition becomes better.
To conclude, as the free space propagation loss in mmWave band is severely high, random beamforming based channel sounding/measurement that radiates towards all directions is inherently power inefficient. The proposed navigation information assisted beam training with constrained sensing range strikes a good balance between randomness and beamforming gain and is an accurate and fast beam training scheme for UAV mmWave communications.

\section{Conclusion}

In this paper, jittering effects analysis and beam training design for UAV mmWave communications with jitter are performed.  According to the geometric relationship, we model UAV mmWave channel with jitter by extracting the relationship between UAV attitude \& position and AoA \& AoD and analyze the jittering effects on UAV mmWave channel response. Then, based on the extracted relationship, we propose to obtain a rough estimate of AoA \& AoD using the information provided by UAV navigation system. Finally, we propose a direction-constrained beam training scheme to refine the estimation of AoA/AoD at UAV side through restricting channel measurement to a narrow angle range that is centered at the rough estimate of AoA/AoD. Numerical results show that our proposed UAV beam training scheme assisted by navigation information can achieve better accuracy  with reduced training length in AoA/AoD estimation and  is thus more suitable for UAV mmWave communications under jittering effects.

\begin{appendices}
\section{Expressions of the partial derivatives in \eqref{FineSearch}}
\begin{align}
&\frac{\partial{g(\Psi_U, \Omega_U)}}{\partial{\Psi_U}} =  2Re\left(\frac{ \mathbf{b}^H(\Psi_U, \Omega_U)\mathbf{M}_U \mathbf{y} \mathbf{y}^H \mathbf{M}_U^H \frac{\partial \mathbf{b}(\Psi_U, \Omega_U)}{\partial \Psi_U}}{\mathbf{b}^H(\Psi_U, \Omega_U)\mathbf{M}_U \mathbf{M}_U^H \mathbf{b}(\Psi_U, \Omega_U)} \right. \notag \\
& \qquad   -  \frac{ \mathbf{b}^H(\Psi_U, \Omega_U)\mathbf{M}_U \mathbf{y} \mathbf{y}^H \mathbf{M}_U^H  \mathbf{b}(\Psi_U, \Omega_U) }{(\mathbf{b}^H(\Psi_U, \Omega_U)\mathbf{M}_U \mathbf{M}_U^H \mathbf{b}(\Psi_U, \Omega_U))^2} \cdot \notag \\
& \qquad\qquad \left. \mathbf{b}^H(\Psi_U, \Omega_U)\mathbf{M}_U \mathbf{M}_U^H \frac{\partial \mathbf{b}(\Psi_U, \Omega_U)}{\partial \Psi_U}  \right) \label{PsiUpartial}
\end{align}
where
\begin{align}
\frac{\partial \mathbf{b}(\Psi_U, \Omega_U)}{\partial \Psi_U} =   \big(\mathbf{v}(\Psi_U, N_{U,x})\odot \boldsymbol{\zeta}_{\Psi_U}\big) \otimes  \mathbf{v}(\Omega_U, N_{U,y})
\end{align}
and $\boldsymbol{\zeta}_{\Psi_U} = [0, j\pi, \cdots , j \pi (N_{U,x} - 1)]^T$.

\begin{align}
&\frac{\partial{g(\Psi_U, \Omega_U)}}{\partial{\Omega_U}} =  2Re\left(\frac{ \mathbf{b}^H(\Psi_U, \Omega_U)\mathbf{M}_U \mathbf{y} \mathbf{y}^H \mathbf{M}_U^H \frac{\partial \mathbf{b}(\Psi_U, \Omega_U)}{\partial \Omega_U}}{\mathbf{b}^H(\Psi_U, \Omega_U)\mathbf{M}_U \mathbf{M}_U^H \mathbf{b}(\Psi_U, \Omega_U)} \right. \notag \\
& \qquad   -  \frac{ \mathbf{b}^H(\Psi_U, \Omega_U)\mathbf{M}_U \mathbf{y} \mathbf{y}^H \mathbf{M}_U^H  \mathbf{b}(\Psi_U, \Omega_U) }{(\mathbf{b}^H(\Psi_U, \Omega_U)\mathbf{M}_U \mathbf{M}_U^H \mathbf{b}(\Psi_U, \Omega_U))^2} \cdot \notag \\
& \qquad\qquad \left. \mathbf{b}^H(\Psi_U, \Omega_U)\mathbf{M}_U \mathbf{M}_U^H \frac{\partial \mathbf{b}(\Psi_U, \Omega_U)}{\partial \Omega_U}  \right) \label{OmegaUpartial}
\end{align}
where
\begin{align}
\frac{\partial \mathbf{b}(\Psi_U, \Omega_U)}{\partial \Omega_U} =   \mathbf{v}(\Psi_U, N_{U,x}) \otimes  \big(\mathbf{v}(\Omega_U, N_{U,y})\odot \boldsymbol{\zeta}_{\Omega_U}\big)
\end{align}
and $\boldsymbol{\zeta}_{\Omega_U} = [0, j\pi, \cdots , j \pi (N_{U,y} - 1)]^T$.

\end{appendices}

\bibliographystyle{IEEEtran}%

\bibliography{bibfile}

\end{document}